\documentclass{vldb}
\usepackage{graphicx}
\usepackage{balance}  
\usepackage{wrapfig}
\usepackage{booktabs}
\usepackage{subcaption}
\usepackage[linesnumbered,vlined,boxed,commentsnumbered,ruled,norelsize]{algorithm2e}
\newtheorem{example}{Example}
\newtheorem{definition}{Definition}

\newtheorem{theorem}{Theorem}
\vldbTitle{Network Embedding: on Compression and Learning}
\vldbAuthors{Esra Akbas and Mehmet Aktas}
\vldbDOI{https://doi.org/10.14778/xxxxxxx.xxxxxxx}
\vldbVolume{12}
\vldbNumber{xxx}
\vldbYear{2019}

\emergencystretch=\maxdimen
\usepackage{wrapfig}
\usepackage{booktabs}
\usepackage{subcaption}
\usepackage{lipsum}
\usepackage{graphicx}
\usepackage{amsmath}
\usepackage{cite}

\usepackage{amsfonts}
\usepackage{amssymb}
\usepackage{MnSymbol}
\usepackage{tikz}
\usepackage{rotating}
\usetikzlibrary{positioning}
\usepackage{array}
\usepackage{enumerate}
\usepackage{hyperref}
\usepackage{float}
\usepackage{multirow}
\usepackage{diagbox}
\newcolumntype{P}[1]{>{\centering\arraybackslash}p{#1}}
\usepackage[pagewise]{lineno}
\usepackage{makecell}
\usepackage[linesnumbered,vlined,boxed,commentsnumbered,ruled,norelsize]{algorithm2e}
\usepackage{colortbl}

\usepackage{booktabs}
\definecolor{Gray}{gray}{0.65}

\begin{document}


\title{Network Embedding: on Compression and Learning}



%
%
%
%

\numberofauthors{2} 

\author{
%
%
\alignauthor
Esra Akbas\\
       \affaddr{Department of Computer Science}\\
       \affaddr{Oklahoma State University}\\
       \affaddr{Stillwater, OK 74078, USA}\\
       \email{eakbas@okstate.edu}
\alignauthor
 Mehmet Aktas\\
       \affaddr{ Department of Mathematics and Statistics}\\
       \affaddr{University of Central Oklahoma}\\
       \affaddr{Edmond, OK 73034, USA}\\
       \email{maktas@uco.edu}
}

\date{30 July 1999}
\newcommand{\my}{\textsf{NECL}}
\newcommand{\dw}{\textsf{DeepWalk}}
\newcommand{\ntv}{\textsf{Node2vec}}
\maketitle

\begin{abstract}
Recently, network embedding that encodes structural information of graphs into a vector space has become popular for network analysis. Although recent methods show promising performance for various applications, the huge sizes of graphs may hinder a direct application of existing network embedding method to them. This paper presents \my, a novel efficient \textbf{N}etwork \textbf{E}mbedding method with two goals.
1) Is there an ideal \textbf{C}ompression of a network? 2) Will the compression of a network significantly boost the representation \textbf{L}earning of the network? For the first problem, we propose a neighborhood similarity based graph compression method that compresses the input graph to get a smaller graph without losing any/much information about the global structure of the graph and the local proximity of the vertices in the graph. For the second problem, we use the compressed graph for network embedding instead of the original large graph to bring down the embedding cost.
\my\ is a general meta-strategy to improve the efficiency of all of the state-of-the-art graph embedding algorithms based on random walks, including \dw\ and \ntv, without losing their effectiveness. Extensive experiments on large real-world networks validate the efficiency of \my\ method that yields an average improvement of 23 - 57\% embedding time, including walking and learning time without decreasing classification accuracy as evaluated on single and multi-label classification tasks on real-world graphs such as DBLP, BlogCatalog, Cora and Wiki.
\end{abstract}
\section{Introduction}
Many real-world data can be modeled as networks to capture the interaction (i.e. edges) between individual units (i.e. vertices). Node classification, community detection and link prediction are some applications of network analysis in many different areas such as social networks and biological networks. Node classification is to find the label of vertices using the topology of the network and other labeled vertices such as predicting demographic values, interest, beliefs or other characteristics
of the user in a social network or prediction labels of proteins in a biological network~\cite{letovsky2003predicting,Gligor_pro, bhagat2011node}. Similarly, link prediction is to determine whether there is an edge between a pair of vertices in a network such as collaboration recommendation on academic social networks and identifying hidden interactions in a protein-protein interaction (PPI) network as a biological network ~\cite{lopes2010collaboration,sharan2007network}.

On the other hand, there are some challenges in network analysis such as high computational complexity, low parallelizability and inapplicability of machine learning methods~\cite{cui2018survey}. Recently, network embedding that encodes structural information of graphs into a vector space has become popular for network analysis~\cite{Zhang2017,hamilton2017representation,goyal2018graph,cai2018comprehensive,cui2018survey}. The network embedding is defined as mapping the network data into a low-dimensional vector space which can capture characteristics or role of vertices in the network based on their connections~\cite{deepwalk}.

Previous researchers considered the network embedding as a dimensionality reduction~\cite{Belkin_lap}. While these methods are effective on small graphs, scalability is the major concern as the time complexity of these methods are at least quadratic in the number of graph vertices. This makes them impossible to apply on large-scale networks with billions of vertices~\cite{Zhang2017,cai2018comprehensive,cui2018survey}. In recent years, the network embedding problem has been changed as a part of the optimization problem to preserve the local and global network structures and node proximity. Researchers focus on the scalable methods that use graph factorization or neural networks. Many of them aim to preserve the first and second order proximity~\cite{tang2015line} or local neighborhood proximity with path sampling using short random walks such as \dw~\cite{deepwalk} and \ntv~\cite{node2vec}. The idea for path sampling is that vertices in a similar neighborhood will get similar paths and so their representation will be similar.

\begin{figure*}[t!]
    \centering
    \begin{subfigure}{0.5\textwidth}
        \centering
        \includegraphics[width=0.99\textwidth]{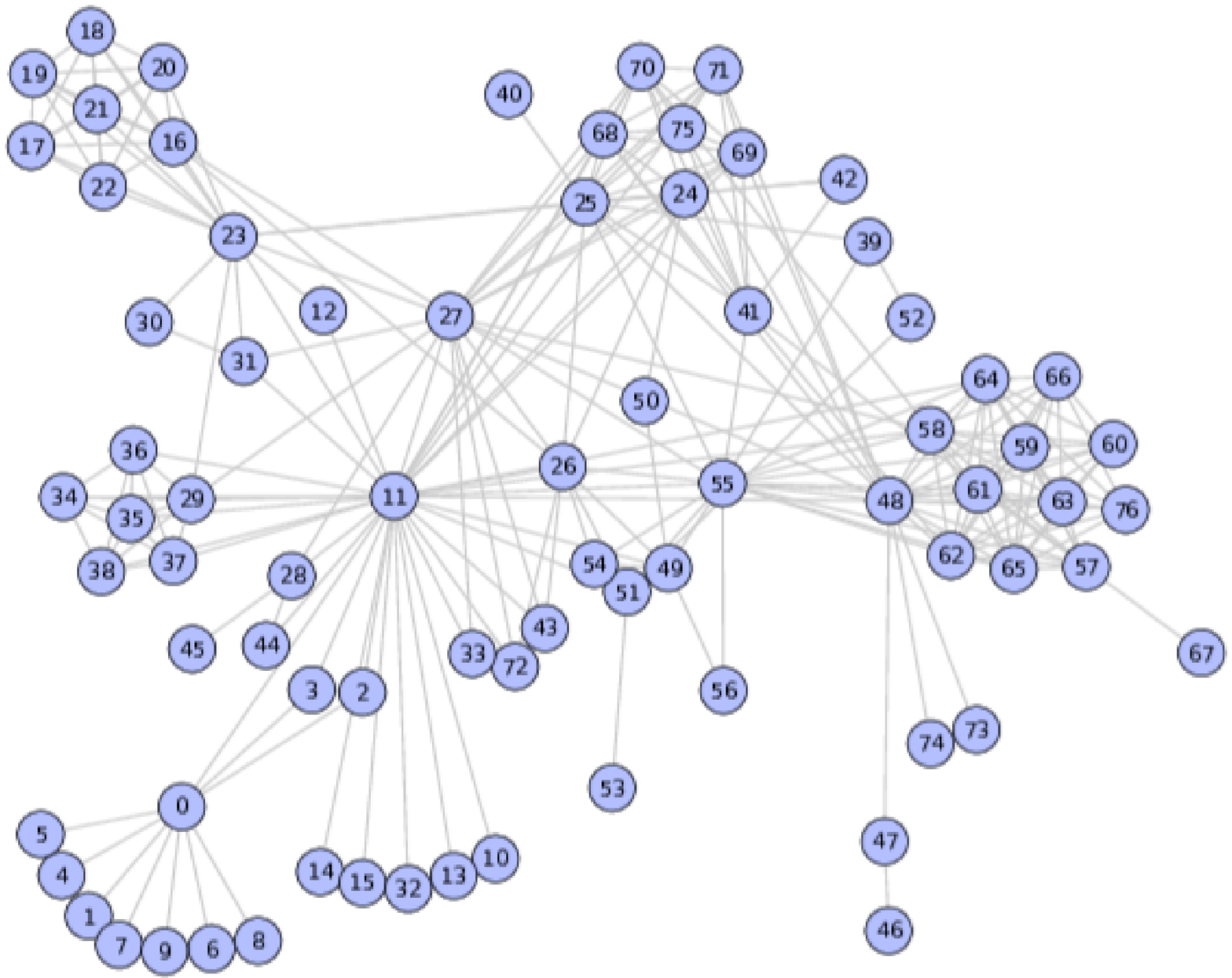}
            \caption{Original network}
   \end{subfigure}%
~  
   \begin{subfigure}{0.5\textwidth}
       \centering
       \includegraphics[width=0.99\textwidth]{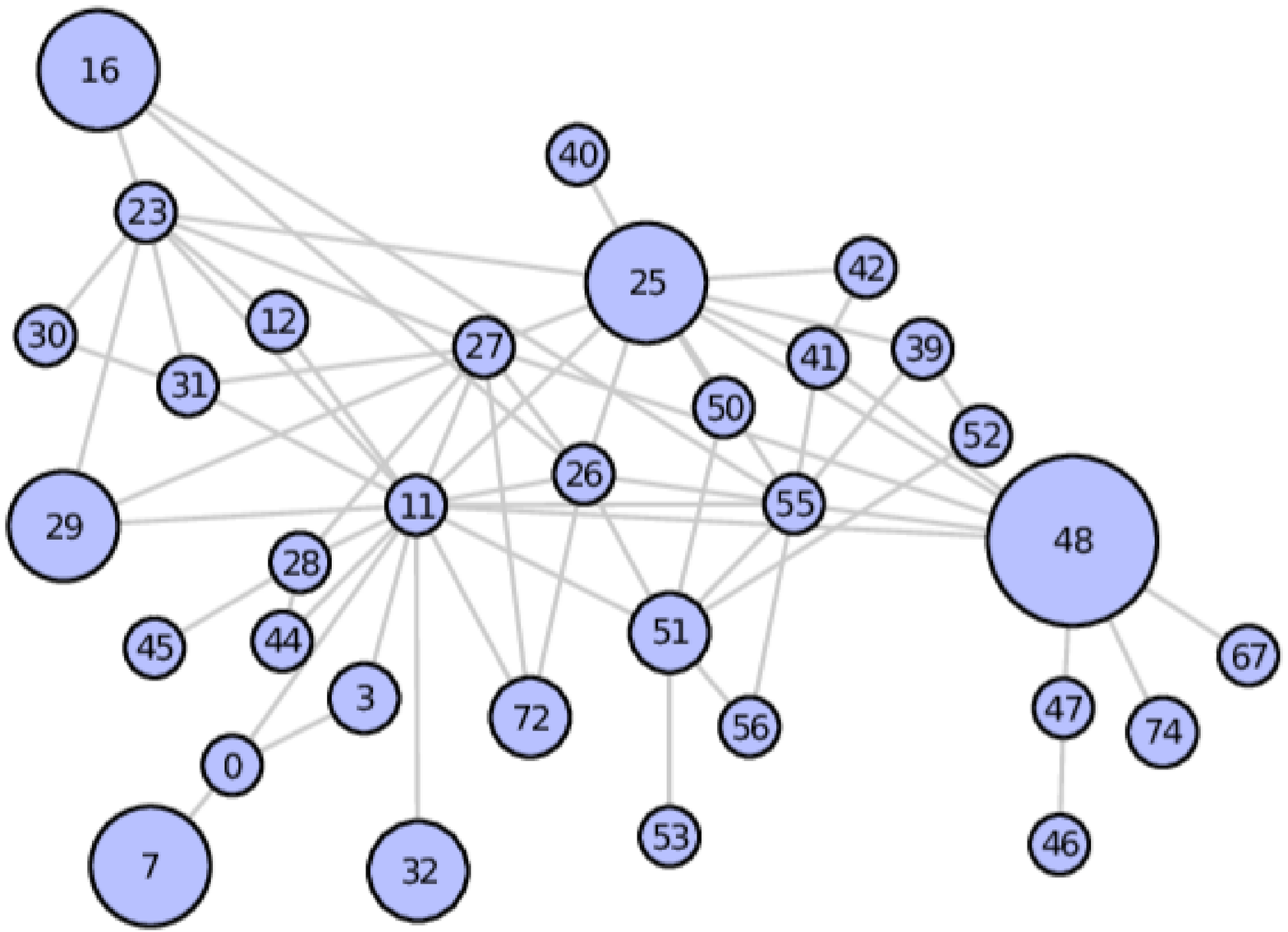}
       \vspace{2mm}
       \caption{Compressed network}
   \end{subfigure}
   \caption{Example of graph compressing on Les Miserables network} 
   \label{fig:toy}
\end{figure*}
Although recent methods show promising performance for various applications, the problem of graph embedding still have some challenges that the huge sizes of real-world graphs may obstruct direct applications of existing graph embedding methods on them. On the other hand, when we consider a compressed or summary graph conserving the key structures and patterns of the original graph, many methods would be applicable to large graphs~\cite{Liu_2018}. The aim of graph compressing is to create a smaller graph without losing any/much information about global structure of the graph and the local relationship between the vertices of the graph~\cite{zhou_comp}. Vertices with similar characteristics are grouped and represented by super-nodes in a compressed graph.

Meanwhile, we have an observation that if two vertices share many common neighbors, they have strong second-order similarities and their paths from random walks will be very similar. From similar paths, we may get very similar representations for these vertices. This means we repeat the same walking and learning process to get similar results for these two vertices. 

In addition to these, optimization on the co-occurrence probability of the vertices could easily get stuck at a bad local minima as the result of poor initialization. This may cause in generating dissimilar representation for vertices within the same or similar neighborhood set. With combining them into super-nodes, we can give initial knowledge to learning process which can result in better representation.

According to these observations, we investigate network embedding via two problems:
\begin{enumerate}
   \item Is there an ideal \textbf{compression} of a network?
  \item  Will the compression of a network significantly boost the representation \textbf{learning} of the network?
  \end{enumerate}
  
As a solution to these problems, we propose \my, a novel network embedding method. For the first problem, we propose a neighborhood similarity based graph compression method that compresses the input graph to get a relatively smaller graph without losing any/much information about global structure of the graph and local proximity of the vertices in the graph. \my\ compresses the graph by merging vertices with high number of similar neighbors into super-nodes. For the second problem, we use the compressed graph for network embedding instead of original large graph to bring down the embedding cost. This benefits the efficiency greatly since we do not need to process similar vertices separately to get similar representation. Instead, we will learn the representation of super-nodes and use their representations as the representation of vertices which are merged to create those super-nodes. Embedding a compressed graph will be easier and more efficient embedding than the original graph. The reason is that we will get less pairwise relationships from random walks on smaller set of super-nodes and this generates less diverse training data for embedding part which makes optimization easier. \my\ is a general meta-strategy to improve the efficiency of the state-of-the-art algorithms for embedding graphs, including \dw\, and \ntv. 

\begin{example}
In Figure~\ref{fig:toy}, we present a graph compressing on the well-known Les Miserables network where vertices correspond to the characters in the novel and edges connect co-appearing characters. While the original network has 77 vertices and 254 edges, the compressed network has 33 vertices and 64 edges. As we see in the figure, the compressed network preserves the local structure of vertices in super-nodes without losing the global structure of the graphs. For example, in Figure~\ref{fig:toy}-(a) neighborhood sets of the vertices $\{1,4,5,6,7,8,9\}$ are same including just node $0$. Hence, random walks from these vertices will have to go from node $0$ and get the same results. We also expect that representations of these vertices should be same or very similar. Another example is that neighborhood set of the vertices $\{16,17,18,19,20,21,22\}$ are same including $\{16,17,18,19,20,21,22,23\}$ except 16 has neighbors $\{26,27\}$. Thus, random walks from these vertices will return to themselves or go far in the graph from node 23. Therefore, they will get the same walking results and as a result similar representations. Instead of walking separately from each of these vertices and learning to get the same or similar feature vectors for them, when we merge them into super-nodes as $7$ and $16$ respectively in the compressed graph in Figure~\ref{fig:toy}-(b), we just need to do walking for one super-node and learn one feature vector that we can use for all of them. After applying merge operation to the whole graph, we get significantly smaller graph (Figure~\ref{fig:toy}(b)) than original graph (Figure~\ref{fig:toy}(a)). Walking on the smaller graph and learning representation from walking results will be more efficient than doing them on the large original graph without decreasing the effectiveness of the learning process.
\end{example}

We summarize the contributions of \my\ as follows,
\begin{itemize}
    \item New graph compressing method: Based on the observation that vertices with similar neighborhood sets get similar results from random walks and eventually similar representation. We merge these vertices into super-nodes to get a compressed (smaller) graph which preserves the characteristics of the original graph.
    
    \item Efficient graph embedding on compressed graph: We do random walks on and embedding the compressed graph, which has less number of vertices and edges the large original graph, as a result they will be easier and more efficient than walking on and embedding the original graph. We use the representation of super-nodes as the representation of vertices in the original graph.
    
    \item Better efficiency without losing effectiveness: The compressed graph preserve the global structure of the network and super node of it preserves local neighborhood of vertices. Using embedding of compressed graph does not decrease the effectiveness. We demonstrate that \my(DW), and \my(N2V) embeddings consistently have better efficiency with less walking and training time but similar or better accuracy than the original methods on multi class and multi-label classification tasks on several real-world networks.

\end{itemize}
\section{Network Embedding using similarity based Compression}
 In this section We first give preliminary information about network embedding and graph compressing, then we describe our neighborhood similarity based graph compression algorithm and how to use compressed graph towards optimizing the efficiency of network embedding.  

\subsection{Preliminaries}
In this section, we briefly discuss the necessary preliminaries for our new meta-strategy for graph embedding. 

In this paper, we consider an undirected, connected, simple graph $G = (V_G;E_G)$ where $V_G$ is the set of vertices, and $E_G \subseteq\{V_G \times V_G\}$ is the set of edges. The set of neighbors for given a vertex $v\in V_G$ is denoted as $N_G(v)$, where $N_G(v)=\{u|u\in V_G:(u,v) \in E_G\}$.\vspace{2mm}

\noindent\textbf{Compressed graph}.\\ Compressed graph of a given graph $G = (V_G;E_G)$ is represented as $CG = (S;M)$ where $S = (V_S;E_S)$ is the graph summary with super-nodes $V_S$ and super-edges $E_S$. Every node $v$ in $V_G$ belongs to a super-node in $V_S$ and $M$ is a mapping from each node $v$ to its super-node in $V_S$. A super-edge $E = (V_i; V_j)$ in $E_S$ represents the set of all edges between vertices in the super-nodes $V_i$ and $V_j$ . \vspace{2mm}

\noindent\textbf{Network Embedding}.\\
\dw~\cite{deepwalk} is the pioneer work that uses the idea of word representation learning~\cite{mikolov2013efficient,mikolov2013distributed} for network embedding. While vertices in a graph are considered as words, neighbors are considered as their context in natural language. A graph is represented as a set of random walk paths sampled from it. The learning process leverages the co-occurrence probability of the vertices that appear within a window in a sampled path. The node representation is learned by training the Skip-gram model~\cite{mikolov2013efficient,mikolov2013distributed} on the random walks. With co-occurrence of the node pairs in the sampled path, a ``corpus" $D$ is generated. To be formal, the corpus $D$ is a multiset that counts the multiplicity of vertex-context pairs. Node pairs with high co-occurrence probability are regarded as neighbors. As the size of the window is usually no less than two, we call these kind of neighbors as higher-order proximity. 

We define a representation as a mapping $\phi : V \rightarrow \mathbb{R}^{d}, d<< |V|$ which represents each vertex $v \in V$ as a point in a low dimensional space $\mathbb{R}^{d}$. Here $d$ is a parameter specifying the number of dimensions of our feature representation. For every source node $u \in V$, we define $N_S(u) \subset V$ as a network neighborhood of node $u$ generated through a neighborhood sampling strategy $S$.

We seek to optimize the following objective function, which maximizes the log-probability of observing a network neighborhood $N_S(u)$ for a node $u$ conditioned on its feature representation, given by $\phi$

\begin{equation}
    \max_f \sum_{u\in V} log Pr(N_S(u)|\phi(u))
    \label{eq:opt}
\end{equation}

There is an assumption as the conditional independence of vertices to make the optimization problem tractable with ignoring the vertex ordering in the random walk. Therefore, the likelihood is factorized by assuming that the likelihood of observing a neighborhood node is independent of observing any other neighborhood node given the feature representation of the source:
$$Pr(N_S(u)|\phi(u))= \prod_{n_i \in N_S(u)} Pr(n_i|\phi(u))$$

The conditional likelihood of every source-neighborhood node pair is modeled as a softmax unit parametrized by a dot product of their features:

 $$Pr(n_i|\phi(u))=\frac{\exp(\phi (n_i) \cdot \phi(u))}{\sum_{v\in V} \exp(\phi (v) \cdot \phi(u))}$$

It is too expensive to compute the summation over all vertices for large networks and we approximate it using negative sampling \cite{mikolov2013distributed}. We optimize Equation (\ref{eq:opt}) using stochastic gradient ascent over the model parameters defining the embedding $\phi$.\vspace{2mm}
 
\noindent\textbf{Random walk based sampling}\\
The neighborhoods $N_S(u)$ are not restricted to just immediate neighbors but can have vastly different structures depending on the sampling strategy $S$. There are many possible neighborhood sampling strategies for vertices as a form of local search. Different neighborhoods coming from different strategies result in different learned feature representations. For scalability of learning, random walk based methods are used to capture the structural relationships of vertices. They maximize the co-occurrence probability of subsequent vertices within a fixed length window of random walks to preserve higher-order proximity between vertices. With random walks, networks are represented as a collection of vertex sequence. In this section, we take a deeper look at the network neighborhood sampling strategy based on random walks and the proximity captured by random walks.

The co-occurrence probability of node pairs depends on the transition probabilities of vertices. Considering a graph $G$, we define adjacency matrix $A$ that is symmetric for undirected graphs. For an unweighted graph, we have $A_{ij}=1$ if and only if there exists an edge from $v_i$ to $v_j$ and $A_{ij}=0$ otherwise. For a graph with adjacency matrix $A$, we define the diagonal matrix, known as degree matrix, as $D_{ij}=\sum_k A_{ik}$ if $i=j$ and $D_{ij}=0$ otherwise. In a random walk, transition probability from one node to other depends on the degree of the vertices. The probability of leaving a node from one of its edges is split uniformly among the edges. We define this 1 step transition probability as $T$: $T=D^{-1}A$ where $T_{ij}$ is the probability of a transition from vertex $v_i$ to vertex $v_j$ within one step. 

 
We observe here that if two vertices, $i,j$, of a graph have many common neighbors, they also have similar transition probabilities to other vertices. This means that if $A_i$ and $A_j$ are similar, $T_i=A_i*D_{ii}^1$ and $T_j=A_j*D_{jj}^1$ will be similar as well. Hence they have similar neighborhood and get similar neighborhood sets from random walks and, as a result, they get very similar representations from the learning process. Therefore, while random walk based neighborhood sampling strategy captures the higher order proximity within the neighborhood of the vertices, the representation learning process based on language model, e.g. Skip-gram~\cite{mikolov2013efficient}, captures the co-occurrence probability of the vertices that appear within a window in a random walk.
 
\subsection{Neighborhood Similarity based Graph Compression}
The critical problem for graph compressing with preserving global structures of the graph is to accurately identify vertices that have similar neighborhood so are more likely to have similar representation. In this section, we discuss how to select vertices to merge into super-nodes.
 
\subsubsection{Motivation} 

The motivation of our method is that if two vertices have the same neighbors, their representations should be very similar. For example, in the toy graph in Figure~\ref{fig:toygc}, the neighbor sets of node $a$ and $b$ are same. Hence, their transition probabilities to the other neighbor vertices are also same, i.e. $p(n_i|a)=p(n_i|b)=1/4$ for all $i\in \{1,2,3,4\}$. Starting on either $a$ or $b$ will yield the same walk, and we will get the same neighborhood set for them. Therefore, instead of walking and learning representations for both $a$ and $b$, it is enough to learn for just one of them. We can merge this node pair $(a,b)$ into one super-node $ab$. Transition probabilities of this super-node to neighbors of $a$ and $b$ are still same with $a$ and $b$ i.e. $p(n_i|ab)=1/4$ for all $i \in \{1,2,3,4\}$. When we obtain the representation of the super-node $ab$, we can use it as the representation of each node in this pair. Merging these vertices keeps preserving the first and second order proximity. Thus this does not affect the results of walking and learning whereas it increases the efficiency.

Furthermore, compressing may change the transition probability of vertices since the number of their neighbors may decrease. As a result, the transition probability of each neighbors changes. For example, in the toy graph in Figure~\ref{fig:toygc}-(a), while the transition probability from $n_1$ to its neighbors is $\frac{1}{|N(n_1)|}$, after compressing, it becomes $\frac{1}{|N(n_1)|-1}$ since number of neighbors decrease by one. In order to avoid this problem, we give weights to edges of super-nodes based on the number of merged edges within the compression. For example, the super-edge between super-node $ab$ and $n_1$ includes 2 edges which are $(a,n_1)$ and $(b,n_1)$. Therefore, the weight of the super-edge ($ab, n_1$) should be 2. 

 \begin{figure}[t]
    \centering
    \begin{subfigure}{0.15\textwidth}
        \centering
        \includegraphics[width=0.9\textwidth]{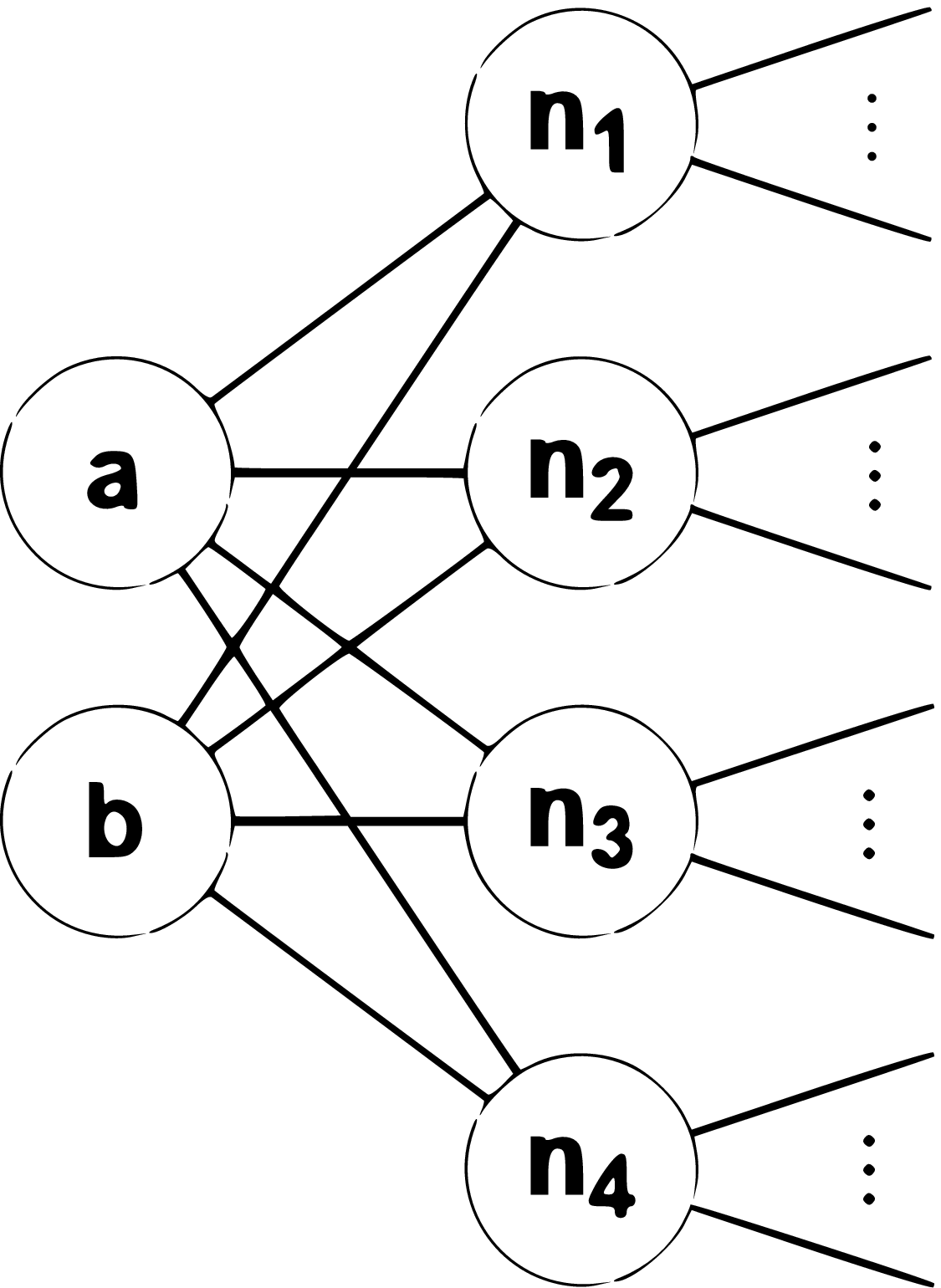}
        \caption{ }
    \end{subfigure}%
    ~
    \begin{subfigure}{0.15\textwidth}
        \centering
        \includegraphics[width=0.9\textwidth]{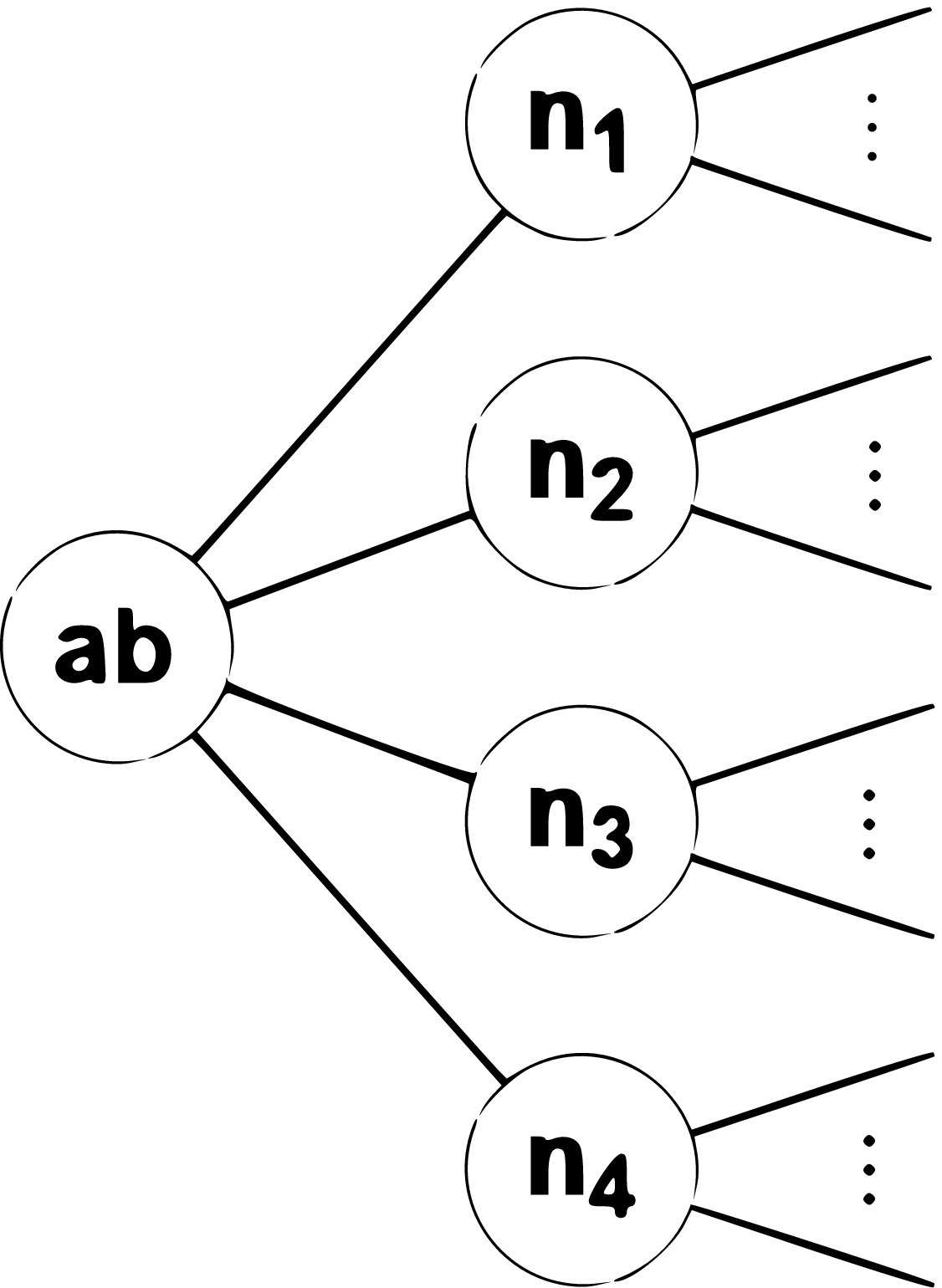}
        \caption{ }
    \end{subfigure}
    \caption{Example of graph compressing: $a$ and $b$ are merged into super-node $ab$ connected to the neighbours of both $a$ and $b$.} 
    \label{fig:toygc}
\end{figure}

In a real-world graph, it is not expected to have too many vertices with the exact same neighbors. However, for many graph mining problems, such as node classification and graph clustering, if two vertices share many common neighbors, they are expected to be in the same class or cluster although their neighbor sets are not completely same. Hence we expect to have similar feature vectors for the vertices in the same class/cluster after embedding. From these observations, we can also apply the same merge operation on these vertices as well. Following the same idea in the example above, if neighbors of two vertices are similar (but not exactly the same), instead of learning representation for each separately, we can merge them into a super-node.

We now define our graph compressing algorithm formally as follows.

\subsubsection{Graph Compressing} For a given graph $G$, if a set of vertices $n_1,n_2,...,n_r$ in $V_G$ have similar neighbors, we merge these vertices into one super-node $n_{12...r}$ to get a smaller compressed graph $G'(V_{G'},E_{G'})$. The compressed graph $G'$ preserves the local and global structure of the original graph but has significantly fewer vertices and edges. 
 
To decide vertices to merge, we define the \textit{neighborhood similarity} based on the transition probability. Before defining the neighborhood similarity, here we first show that cosine similarity between transition probabilities of two vertices $u, v$, $T_u$ and $T_v$, are determined by the number of their common neighbors.

\begin{theorem} \label{thm:nbor}
Let $T$ be the 1-step transition probability matrix of vertices $V$ in a graph $G$ and let $u,v \in V$. Let $T_u$ and $T_v$ be the transition probability from vertices $u$ and $v$ to other vertices. Then the cosine similarity between $T_u$ and $T_v$ is 
$$
sim(T_u,T_v)=|N(u)\cap N(v)|.
$$
 \end{theorem}
 
\begin{proof}
The cosine similarity between $T_u$ and $T_v$ is defined by
\begin{equation}
sim(T_u,T_v)=\frac{\sum_iT_{ui}T_{vi}}{||T_u||  ||T_v||}
\end{equation}

By definition of $T$, we have $T_u=\frac{A_u}{|N(u)|} $ and  $T_v=\frac{A_v}{|N(v)| }$. Furthermore, we have 
$$||T_u||=1/|(N(u)\text{, } ||T_v||=1/|(N(v)$$ and 
$$\sum_i A_{ui} A_{vi}=|N(u) \cap N(v)|.$$ 
Hence, if we plug in these into the Equation (1), we get

\begin{align*} 
\displaystyle 
sim(T_u,T_v) &=\frac{\sum_iT_{ui}T_{vi}}{||T_u||  ||T_v||} =\frac{\sum_i \frac{A_{ui}}{|N(u)|} \frac{A_{vi}}{|N(v)|}}{1/|N(u)| \times 1/|N(v)|} \\
&= \frac{{\frac{1}{|N(u)||N(v)| }}|N(u) \cap N(v)|}{\frac{1}{|N(u)||N(v)| }} \\
&=|N(u)\cap N(v)|.
\end{align*}

This finalizes the proof.
\end{proof}

From Theorem \ref{thm:nbor}, we see that the similarity of transition probabilities from two vertices to other vertices depends on the similarity of their neighbors. Therefore we define the neighborhood similarity between two vertices as follows.

 \begin{definition}{(Neighborhood similarity)} Given a graph $G$, the neighborhood similarity between two vertices $u,v$ is given by 
 \end{definition} \begin{equation}
 Nsim(u,v)=\frac{2|N(u) \cap N(v)|}{|N(u)|+|N(v)|}
 \end{equation}

In order to normalize the effect of high degree vertices, we divide the number of common neighbors by degree of vertices. The neighborhood similarity is between 0 and 1 where it is 0 when two vertices have no common neighbor and 1 when both have the exact same neighbors. According to the neighbor similarity, we merge vertices whose similarity value is greater than a given threshold.

\begin{algorithm}[t]
\DontPrintSemicolon
\SetAlgoLined

\KwIn{$G(V_G,E_G)$, similarity threshold $\lambda$}
\KwOut{$S(V_{G'},E_{G'}, W_E)$,mapping \textbf{M} \\\textbf{M} is a mapping from super-node to original node}
$S\leftarrow G$\;
$NSQ \leftarrow \emptyset $\;
\For{$v \in V_G$}
{\For {$u \in N_G(v)$}
{\For {$k \in N_G(u)$}
{    
    Compute Neighborhood similarity between $v$ and $k$ as $NSim(v,k)$ \\
    $NSQ \leftarrow NSQ \cup (v,k)$
}
}}
\For{ $(v,k) \in NSQ$}
{
\If{$NSim(v,k) >\lambda$}
{Merge them into a super-node $s_{v,k}$\;
$M(s_{v,k}) \leftarrow  v$; $M(s_{v,k}) \leftarrow  k$\;
Delete $v$ and $k$ from $S$ and add $s_{v,k}$ into $S$. }
\For{$ng \in N_S(v)$}
{add edge between $s_{v,k}$ and $ng$ \\
$w(ng,s_{vk})=w(ng,v)$ }
\For{$ng \in N_S(k)$}
{add edge between $s_{v,k}$ and $ng$ if there is no\\
$w(ng,s_{vk})=w(ng,s_{vk})+w(ng,k)$ }
}
\caption{Graph Compressing($G, \lambda$)}
\label{alg:spar}
\end{algorithm}

The neighborhood similarity based graph compressing algorithm is given in Algorithm~\ref{alg:spar}. It is clear that the vertices with a nonzero neighborhood similarity are 2-step neighbors. Therefore, we do not need to compute the similarity between all pairs of the vertices, instead, we just need to compute the similarity between vertices and its neighbors' neighbors. For each node $v \in V_G$, we compute the similarity between $v$ and each $k$ as neighbors of neighbors (Line 3-10). Then, we check the similarity value of all pairs ($u$, $k$) in the list and if it is higher than the given threshold $\lambda$ (line 12), we merge them $u$ and $k$ into a super-node $s_{u,k}$ (line 13). Then we delete edges of $u$ and $k$ and add edges between neighbors of $u$ and $k$ and new super-node $s_{u,k}$ (line 17-24). We give the weights to edges of super-nodes. Original edge weights are assigned to 1. Threshold $\lambda$ decides the trade-off between efficiency and effectiveness. If we use a larger value, it will merge less number of vertices. On the other hand, if we use a smaller value, we merge more vertices and as a side effect, we may merge some dissimilar vertices as well, that results in an increase in efficiency but causes a decrease in accuracy. Note that, the order of merging is arbitrary and one super-node may include more than two vertices of the original graph. For example, if the similarity between the vertices $x$ and $y$, $NSim(x,y)$, and the vertices $y$ an $z$, $NSim(y,z)$, are both bigger than given threshold, we merge $x$ and $y$ in $s_{x,y}$ and then we merge $s_{x,y}$ and $z$ into $s_{x,y,z}$. Therefore, during the merge operation, we check whether the node $y$ is merged with another node and if so, we get the super-node of the original node $x$.

\subsubsection{Network embedding on compressed graph}
Our algorithm for network embedding on a compressed graph is given in Algorithm~\ref{alg:emb}. After getting the weighted compressed graph $S$ (line 1), we obtain the representation of super-nodes $V_S$ as $\phi_s$ in the compressed graph with the provided network embedding algorithm (line 2). We apply any random walk based representation learning algorithm on the compressed graph. We just need to apply weighted random walks to take the edge weights into consideration. As the size of the compressed graph is smaller than the original graph, it is more efficient to get embeddings of super-nodes than vertices. Finally, we assign the embedding of super-nodes to vertices according to the mapping $M$ obtained from the compression (line 3-6).

\begin{algorithm}[t!]
\DontPrintSemicolon
\SetAlgoLined
\KwIn{$G(V_G,E_G)$, similarity threshold $\lambda$}
\KwOut{Representation $\phi(u)$ for all $v \in V_G$}
$S,M\leftarrow$ GraphCompressing($G, \lambda$)\\
$\phi_{S} \leftarrow$ WightedGraphEmbeding($S$)\\
\For{$V_i\in V_S$}{
\For{$v_j \in M(V_i)$}{
$\phi(v_j) \leftarrow \phi_S(V_i)$}}
\caption{\my: Network Embedding on Compressed Graph}
\label{alg:emb}
\end{algorithm}

\section{Experiments}
We perform experimental studies to evaluate the efficiency and effectiveness of our algorithms on challenging multi-class and multi-label classification tasks in several real-world networks. We first provide an overview of the datasets and embedding methods used for experiments. We further show the performance of algorithms and also the improvement of our method on efficiency and discuss parameter sensitivity for different values of similarity threshold $\lambda$ and training ratio.

\subsection{Datasets}
The general statistics of the datasets used for experiments are reported in Table \ref{table:datatcl}.

\begin{itemize}
\item \textbf{Cora} - Cora is a citation network of machine learning papers. The labels of vertices indicate the topic of the paper. Each paper has a single topic. We convert it to an undirected graph and just use link information. We do not consider the attribute information of vertices which are word vectors indicating the absence/presence of the corresponding word from a given dictionary.
\item \textbf{Wiki} - Wiki is a network with vertices as web pages from 19 classes. Each page has a single label. The link among different vertices is the hyperlink on the web page. We convert it to undirected graph and just use link information. We do not consider the attribute information of vertices which are the TF-IDF values of web pages.

\item \textbf{DBLP} - This is a network of co-authorship of researchers in computer science. The labels represent the research areas in which a researcher publishes his work. The 4 research areas included in this dataset are DB, DM, IR, and ML. A researcher may have more than one research area.
    
\item \textbf{BlogCatalog} - BlogCatalog is a social network of users as bloggers on the BlogCatalog website. The link shows the relationships between users. The labels of a user represent the categories that blogger has interest and published in extracted from the metadata provided by the user. A user may have more than one label.

\end{itemize}

\begin{table}[t]
\footnotesize
\centering
\caption{Graphs statistics ($K =10^3$ and $M=10^6$)}
\label{table:datatcl}
\renewcommand{\arraystretch}{1.2}
\begin{tabular}{|c||c|c|c|c|}
 \hline 	
 \cellcolor{gray!60}\textbf{Network} &\cellcolor{gray!60}$\mathbf{|V|}$&\cellcolor{gray!60}$\mathbf{|E|}$ & \cellcolor{gray!60}\textbf{class \#} & \cellcolor{gray!60}\textbf{Multi-label}\\\hline\hline
 Wiki &2405&23192&17&No\\\hline
 Cora & 2708 & 10858 & 7&No\\\hline
 DBLP& 51330 & 133664 & 4&Yes\\\hline
 BlogCatalog& 10312K &668K & 39&Yes  \\\hline
\end{tabular}
\end{table}
\subsection{Baseline methods}
For the performance evaluation, we use \dw\ and \ntv\ as baseline embedding methods in our model and compare our model with them. We combine each baseline methods with \my\ and compare their performance. We give a brief explanation about these methods as follows: 
\begin{itemize}
    \item \textbf{\dw} - \dw\ is a random walk based method for network embedding. It preserves the higher order proximity between vertices with generating random walks of fixed length from all the vertices of a graph. With considering the walks as sentences in a language model, it optimizes the log-likelihood of random walks using the Skip-gram model~\cite{mikolov2013distributed}, which is for learning word embeddings. \dw\ uses hierarchical softmax for the efficiency of optimization.
\item \textbf{\ntv} - \ntv\ is a random walk based network embedding method which makes an improvement to the random walk phase of \dw. It applies biased random walks using the return parameter $p$ and the in-out parameter $q$ to combine DFS-like and BFS-like neighborhood explorations. With this way, they preserve the network community and structural roles of vertices. Different than \dw, \ntv\ uses negative sampling for optimization.
\end{itemize}

\textbf{Parameter Settings:} For \dw\, \ntv\, and \my(DW), \my(N2V), we set the following parameters: the number of random walks $\gamma$, walk length $t$, window size $w$ for the Skip-gram model and representation size $d$. The parameter setting for all models is $\gamma = 40$, $t = 10$, $w = 10$, $d = 128$. The initial learning rate and final learning rate are set to 0.025 and 0.001 respectively in all models.

\subsection{Classification}
In this section, we compare our method with the baseline methods in two different classification tasks, namely single-label and multi-label classifications. In the former case, vertices have only one label (Cora and Wiki datasets) and in the latter case, they can have more than one label.

To evaluate our method, firstly, we obtain the embeddings of the vertices with each method and then use them as features to train a classifier. A portion of the labeled vertices are sampled randomly from the graph to train the classifier and the rest of the vertices are used for testing. 

To have a detailed comparison between \my\ and the baseline methods, we vary the portion of labeled vertices for classification and similarity threshold value $\lambda$ and present the macro and micro $F_1$ scores with walking and embedding times. We also report the number of edges and vertices in the compressed graph to see how much each graph is compressed. We increase $\lambda$ from 0.45 to 0.8 to test its effect on the efficiency and effectiveness of the embedding algorithms. While we vary the training ratio on the Cora, Wiki and DBLP datasets from $1\%$ to $50\%$, we vary the training ratio on the BlogCatalog network from $10\%$ to $80\%$. The number of class labels of BlogCatalog is about 10 times than other graphs, thus we use a larger portion of labeled vertices

To ensure the reliability of our experiment, the classification process is repeated 10 times, and the average macro $F_1$, micro $F_1$ scores and running times are reported. All are performed on a server running Ubuntu 14:04 with $4$ Intel $2.6$ GHz ten-core CPUs and 48 GB of memory.

\subsubsection{Single-label Classification}

In these experiments, each node in the datasets has a single label from multi-class values. For the classification task, the multi-class SVM is employed as the classifier which uses the one-vs-rest scheme.

\begin{table*}

\centering{
\caption{Performance comparison of the single-label classification tasks for the similarity threshold $\lambda=0.5$ and training ratio $5\%$ for Cora and Wiki} 
\label{table:compsin}

{\begin{tabular}{|c | c c c | c c c | }
 \hline
 
 \multirow{2}{*}{$ $} & \multicolumn{3}{|c|}{\cellcolor{gray!60}\textbf{Cora (5\%)}} & \multicolumn{3}{|c|}{\cellcolor{gray!60}\textbf{Wiki (5\%)}} \\ \cline{2-7}
 
 & \cellcolor{gray!25}\textbf{\my(DW)} & \cellcolor{gray!25}\textbf{DW} & \cellcolor{gray!25}\textbf{Gain} \% & \cellcolor{gray!25}\textbf{\my(DW)} & \cellcolor{gray!25}\textbf{DW} & \cellcolor{gray!25}\textbf{Gain} \% \\ \hline
\cellcolor{gray!25}\textbf{Macro} $\mathbf{F_1}$ & 0.671 & 0.675 & -0.61 & 0.344	& 0.342 & 0.58 \\
\cellcolor{gray!25}\textbf{Micro} $\mathbf{F_1}$ & 0.704 & 0.704 & 0.01 & 0.477 & 0.483 & -1.24 \\
\cellcolor{gray!25}\textbf{Time(s)} & 5.17 & 8.29 & \textbf{37.65} & 4.84 & 8.98 & \textbf{46.07}\\
\hline
 & \cellcolor{gray!25}\textbf{\my(N2V)} & \cellcolor{gray!25}\textbf{N2V} & \cellcolor{gray!25}\textbf{Gain} \% & \cellcolor{gray!25}\textbf{\my(N2V)} & \cellcolor{gray!25}\textbf{N2V} & \cellcolor{gray!25}\textbf{Gain} \%\\ \hline
 \cellcolor{gray!25}\textbf{Macro} $\mathbf{F_1}$ & 0.666 & 0.671 & -0.84 & 0.342 & 0.348 & -1.72 \\
\cellcolor{gray!25}\textbf{Micro} $\mathbf{F_1}$ & 0.691 & 0.709 & -2.52 & 0.475 & 0.498 & -4.62\\
 \cellcolor{gray!25}\textbf{Time(s)} & 11.96 & 17.96 & \textbf{33.39} & 9.41& 19.10 & \textbf{50.74} \\ \hline
 & \cellcolor{gray!25}\textbf{Compressed} & \cellcolor{gray!25}\textbf{Original}& \cellcolor{gray!25}\textbf{Gain} \%  & \cellcolor{gray!25}\textbf{Compressed} & \cellcolor{gray!25}\textbf{Original}& \cellcolor{gray!25}\textbf{Gain} \% \\ \hline
\cellcolor{gray!25}$\mathbf{|V|}$ &1427&	2708 & \textbf{47.30} & 1060 & 2405 & \textbf{55.93} \\
 \cellcolor{gray!25}$\mathbf{|E|}$ & 5236 &	10858 & \textbf{51.78} & 8584 & 23192 & \textbf{62.99} \\ \hline
\end{tabular}}}
\end{table*}
\vspace{3mm}
\begin{figure*}
    \centering
    \begin{subfigure}[h]{0.32\textwidth}
        \centering
        \includegraphics[width=1.12\textwidth]{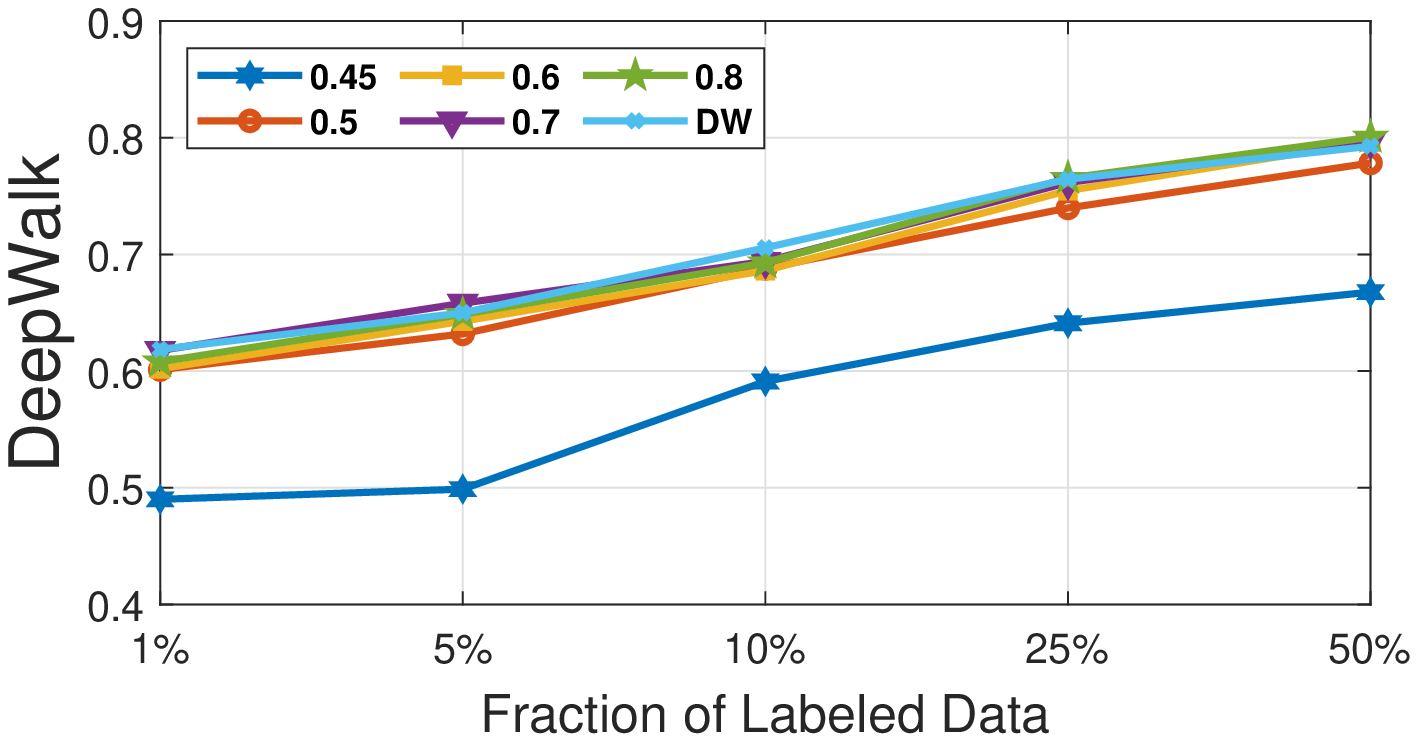}
    \end{subfigure}%
    ~ 
    \begin{subfigure}[h]{0.32\textwidth}
        \centering
        \includegraphics[width=1.12\textwidth]{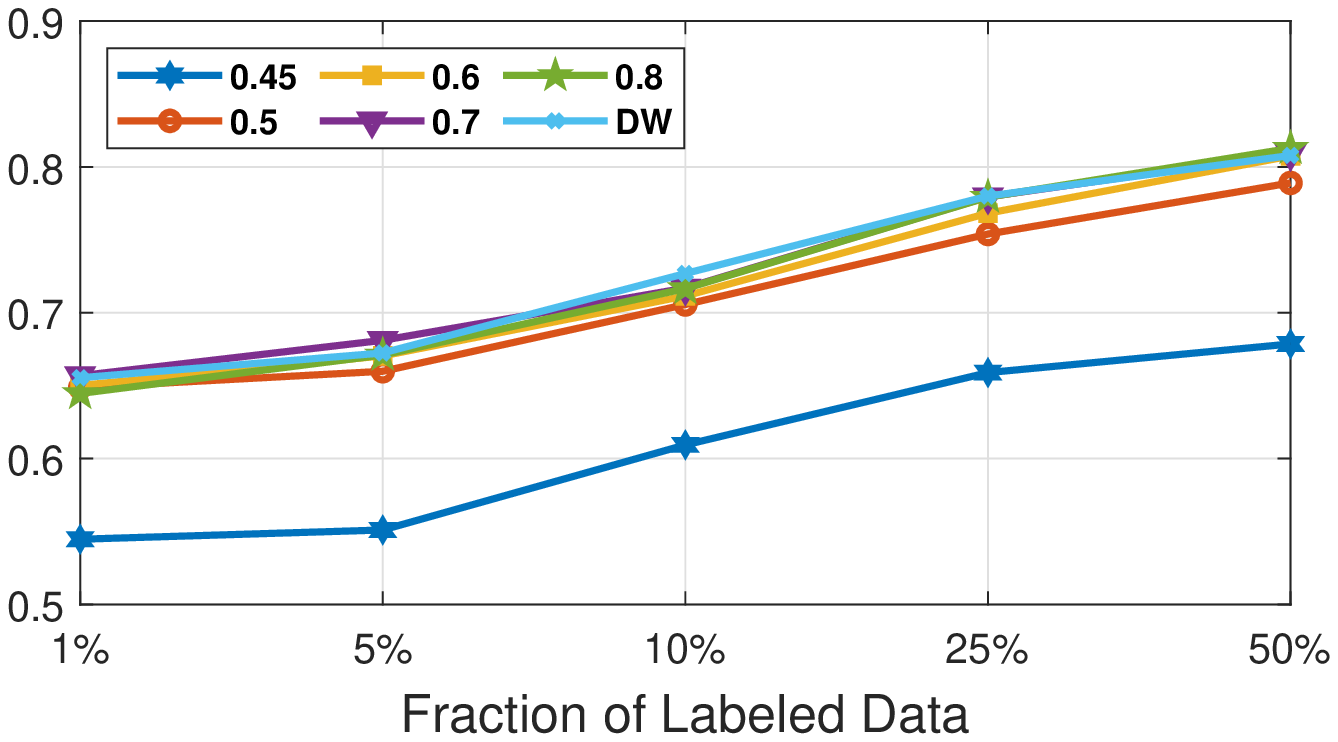}
    \end{subfigure}
        ~ 
    \begin{subfigure}[h]{0.32\textwidth}
        \centering
        \includegraphics[width=1.12\textwidth]{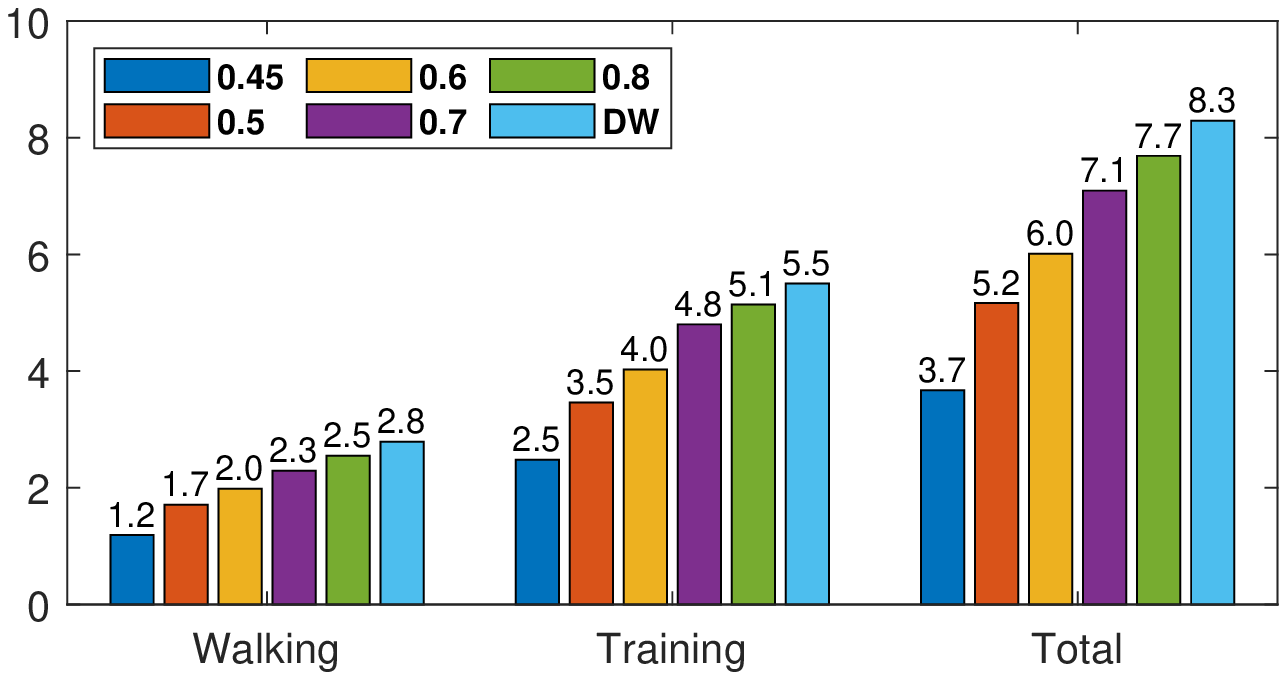}
    \end{subfigure}
~
    \begin{subfigure}[h]{0.32\textwidth}
        \centering
        \includegraphics[width=1.12\textwidth]{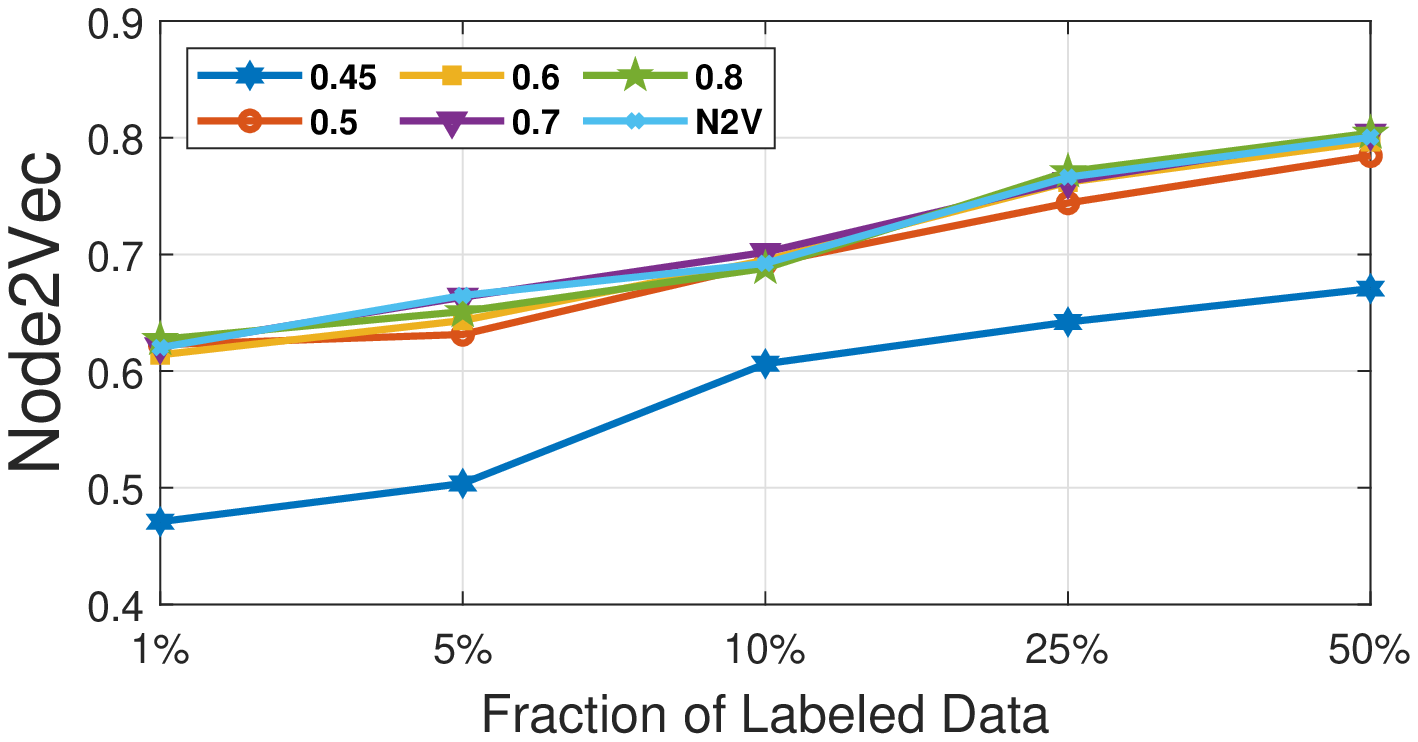}
        \caption{Macro $F_1$}
    \end{subfigure}%
    ~ 
    \begin{subfigure}[h]{0.32\textwidth}
        \centering
        \includegraphics[width=1.12\textwidth]{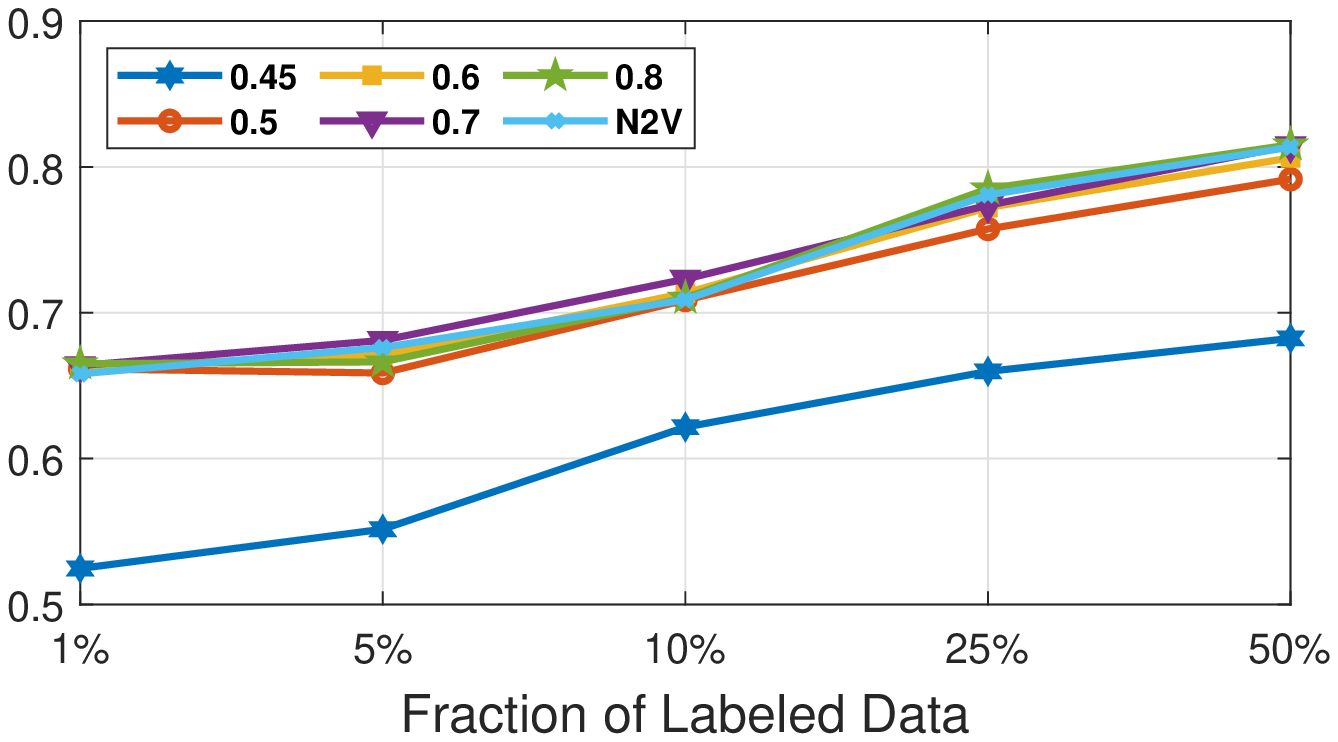}
        \caption{Micro $F_1$}
    \end{subfigure}
        ~ 
    \begin{subfigure}[h]{0.32\textwidth}
        \centering
        \includegraphics[width=1.12\textwidth]{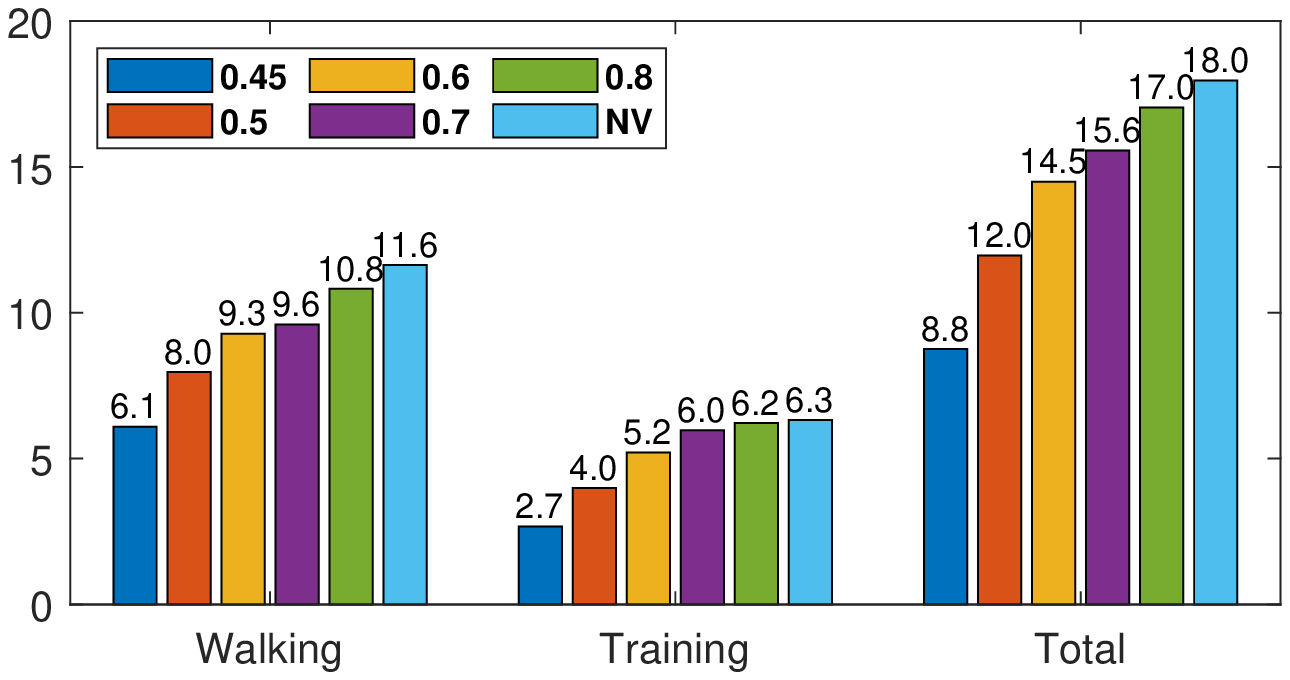}
        \caption{Running Time(s)}
    \end{subfigure}
    
    \caption{Detailed single-label classification result on Cora.}
    \label{fig:core}
\end{figure*}
\vspace{2mm}

\begin{figure*}
    \centering
    \begin{subfigure}[h]{0.32\textwidth}
        \centering
        \includegraphics[width=1.12\textwidth]{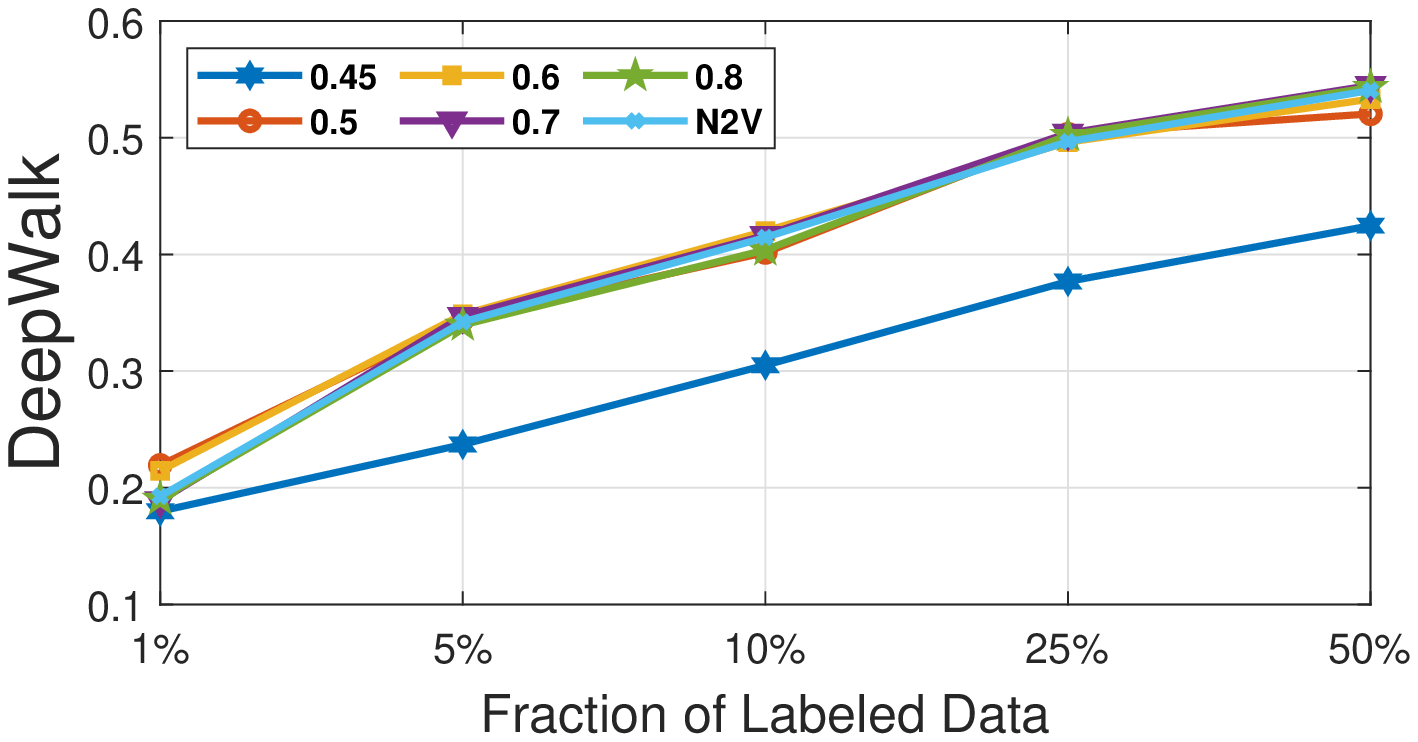}
    \end{subfigure}%
    ~     
    \begin{subfigure}[h]{0.32\textwidth}
        \centering
        \includegraphics[width=1.12\textwidth]{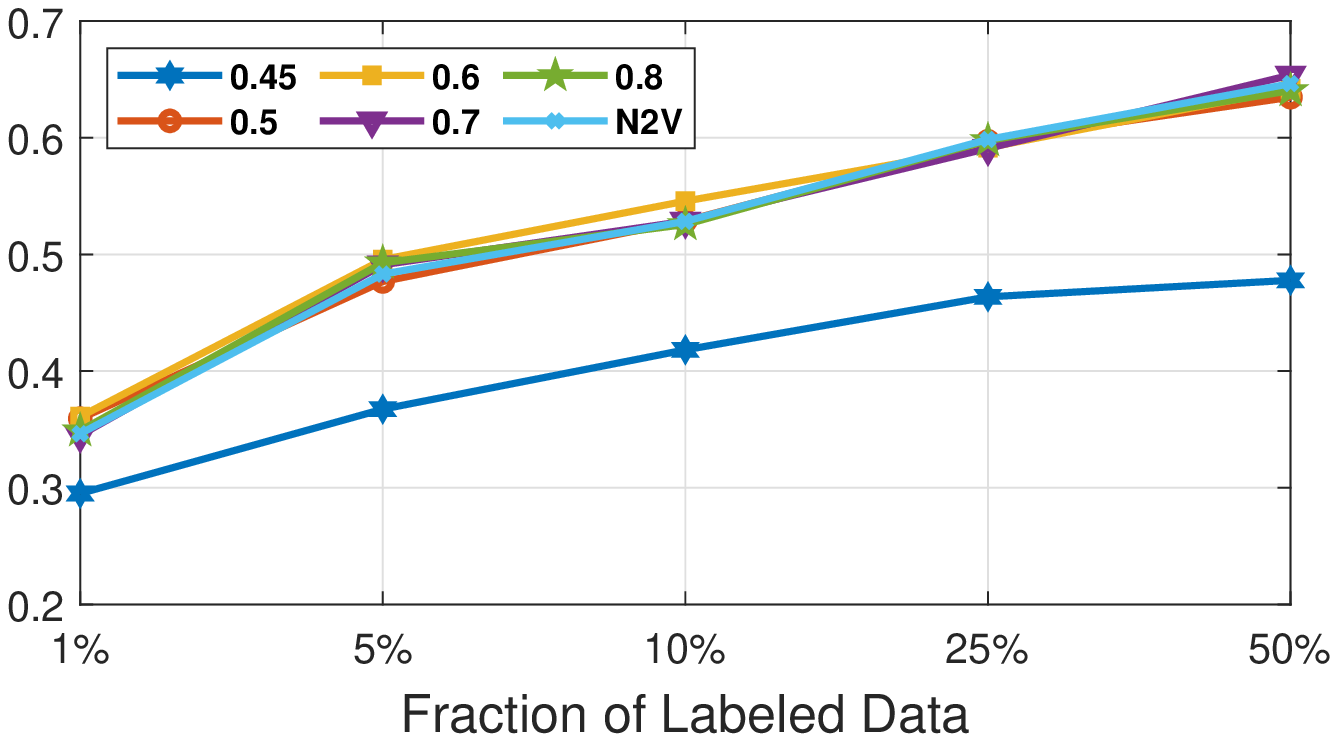}
    \end{subfigure}%
    ~ 
    \begin{subfigure}[h]{0.32\textwidth}
        \centering
        \includegraphics[width=1.12\textwidth]{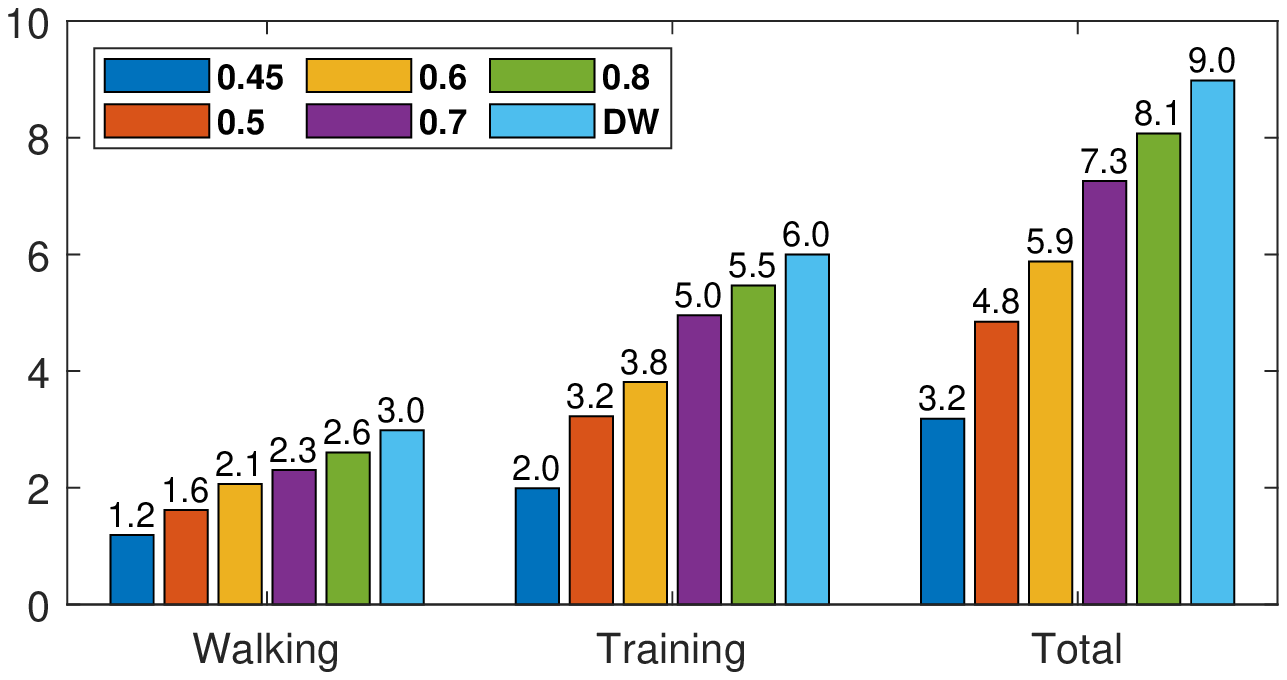}
    \end{subfigure}
        ~ 
    \begin{subfigure}[h]{0.32\textwidth}
        \centering
        \includegraphics[width=1.12\textwidth]{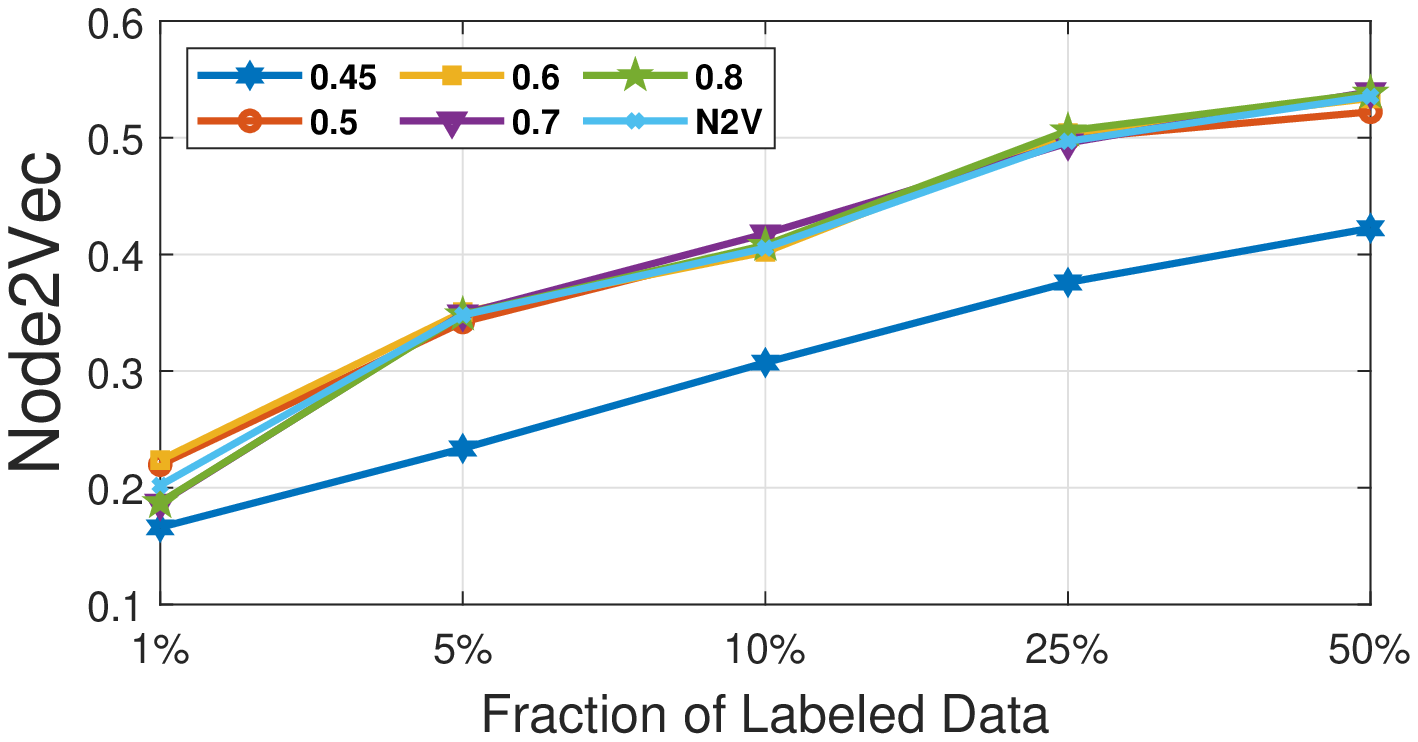}
        \caption{Macro $F_1$}
    \end{subfigure}
~
        \begin{subfigure}[h]{0.32\textwidth}
        \centering
        \includegraphics[width=1.12\textwidth]{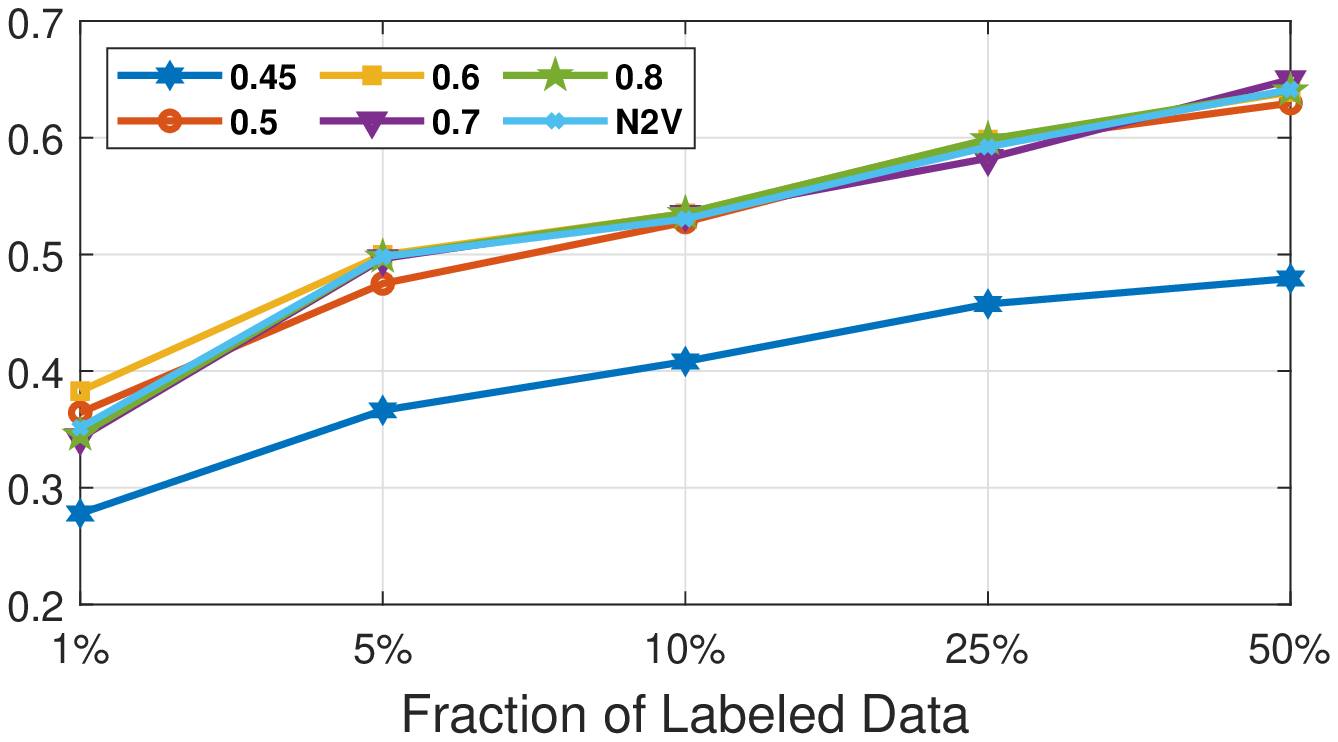}
        \caption{Micro $F_1$}
    \end{subfigure}%
    ~ 
    \begin{subfigure}[h]{0.32\textwidth}
        \centering
        \includegraphics[width=1.12\textwidth]{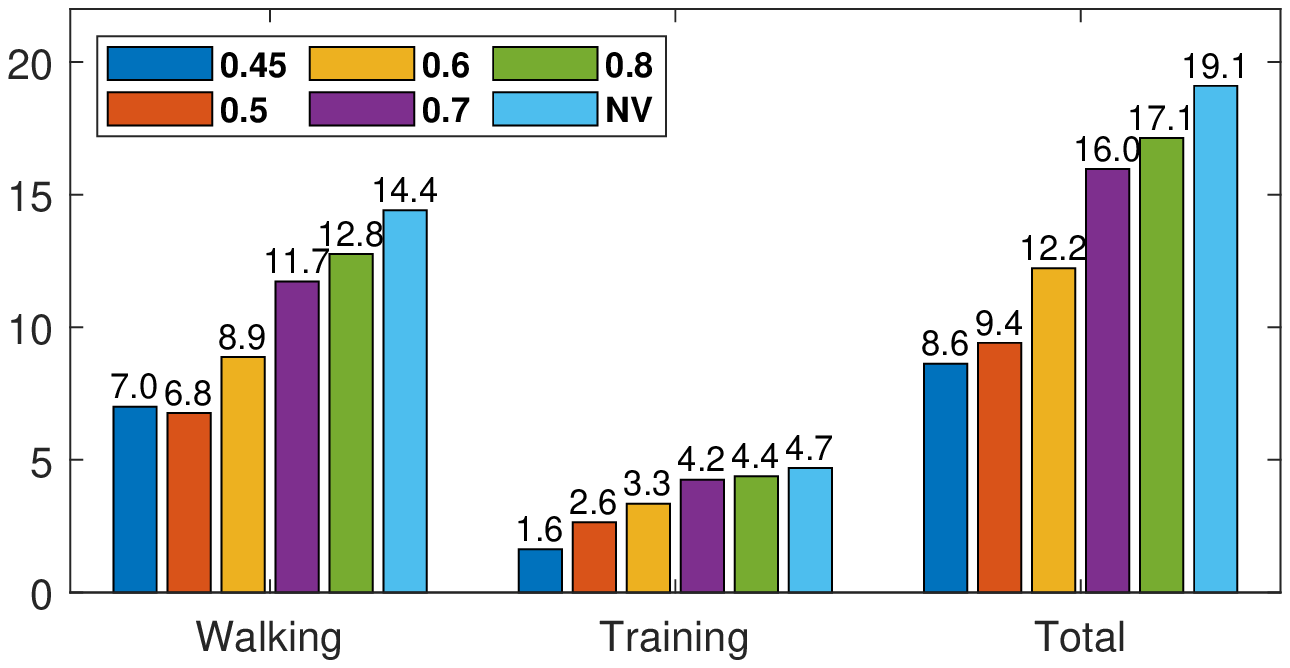}
        \caption{Running Time(s)}
    \end{subfigure}%

    \caption{Detailed single-label classification result on Wiki.}
    \label{fig:wiki}
\end{figure*}

Table~\ref{table:compsin} shows the macro $F_1$ and micro $F_1$ scores, and time for embedding on Cora and Wiki with $5\%$ labeled vertices and $\lambda=0.5$ similarity threshold value. When the similarity threshold $\lambda < 0.5$, graphs become too small and accuracy decrease dramatically. Therefore, we select $\lambda=0.5$ as the cutting point for compression. As we see in the table, for both datasets, while there is no (significant) change on effectiveness as the macro $F_1$ and micro $F_1$ scores, there is a significant gain on efficiency as the total embedding time. While there is around 33.4\% and 37.65\% efficiency improvement on Cora dataset, there is 46\% and 50.7\% efficiency improvement on Wiki when it is compared with base line results, \dw\ and \ntv\ respectively. There is also significant graph compression ratio for both datasets. The size of the graph is decreased significantly with compressing. While the number of vertices is decreased to 1427 from 2708 (47.3\%) for Cora and to 1060 from 2405 (55.9\%) for Wiki, the number of edges is decreased to 5236 from 10858 (51.8\%) for Cora and to 8584 from 23192 (62.9\%) for Wiki.

The detailed comparison between \my\ and the baseline methods with varying the portion of labeled vertices for classification and similarity threshold value $\lambda$ is given in Figure~\ref{fig:core} and Figure~\ref{fig:wiki} for the Cora and Wiki datasets respectively. We report details of embedding time as walking, training and total embedding time separately. As we see from figures, while the macro and micro $F_1$ scores are very similar with or higher than baseline results for $\lambda\geq 0.5$, the running times are significantly different for both datasets. There is an improvement in both walking and training time for embedding. For both datasets, when the similarity threshold $\lambda < 0.5$, the macro $F_1$ and micro $F_1$ scores dramatically decrease since it merges many vertices and edges so this may cause information loss in the graph. 

\subsubsection{Multi-label Classification}
The datasets used in these experiments are multi-labeled, i.e., a node can belong to more than one class. For this task, we train a one-vs-rest logistic regression model with $L_2$ regularization on the graph embeddings for prediction. The logistic regression model is implemented by LibLinear~\cite{fan2008liblinear}.

Table~\ref{table:compare} shows the macro $F_1$ and micro $F_1$ scores, and time for embedding on DBLP and BlogCatalog with 5\% and 50\% labeled vertices respectively and $\lambda=0.5$ similarity threshold value. Similar to single label classification, we select $\lambda=0.5$ as the cutting point for compression. 

As we see in the table, for \textbf{DBLP} dataset, while the macro $F_1$ and micro $F_1$ scores of \my\ are very similar with baseline results, there is a significant gain on embedding time which are 57.46\% and 56.75\% for \dw\ and \ntv\ respectively. There is also a high graph compression ratio for this dataset. While the number of vertices is decreased to 8824 from 29199 (69.8\%), the number of edges is decreased to 32984 from 133664 (75.3\%). 

As a scale-free network with complex structure, \textbf{BlogCatalog} is challenging for graph coarsening. While there is a slight decrease in both macro and micro $F_1$ scores (2.9\% on macro $F_1$ and 5.6\% on micro $F_1$ for \dw\ and 4.1\% on macro $F_1$ and 6.6\% on micro $F_1$ for \ntv), we obtain about 28.2\% and 23.4\% gains in the total running time respectively. Furthermore, we reduce the number of vertices and edges about 17.5\% and 18.6\% percent in the compressed graph respectively.

The detailed comparison between \my\ and the baseline methods with varying the portion of labeled vertices and similarity threshold value $\lambda$ for multi-label classification is given in Figure~\ref{fig:db} and Figure~\ref{fig:bc}. In addition to the macro and micro $F_1$ scores achieved on DBLP and BlogCatalog datasets, we also report detailed embedding time as walking, training and total embedding time separately in Figure~\ref{fig:db}-(c) and Figure~\ref{fig:bc}-(c).

For the \textbf{DBLP} dataset (Figure \ref{fig:db}), as it happens in Cora and Wiki, \my\ has very similar, even slightly higher,\ macro and micro $F_1$ scores than baseline methods for $\lambda \geq 0.5$ at all training ratios, but again the scores decrease dramatically for smaller $\lambda$ values. On the other hand, there is a significant gain in walking, training and total embedding time.

For the \textbf{BlogCatalog} dataset (Figure \ref{fig:bc}), there are similar results as well. Macro $F_1$ scores are close each other for $\lambda\geq 0.5$; however, micro $F_1$ scores are slightly different for $\lambda\leq 0.7$. For the comparison between \my\ and both baseline methods, \dw\ and \ntv, although there is a slight decrease in both macro and micro $F_1$ scores, we obtain gains on the running times, especially on walking times. For \ntv, walking time takes a large portion of the embedding time as a result of thebiased walking. The biggest reason is that, since the degree of vertices is higher, defining a biased probability on them takes longer time. 

In short, for both the single-label and the multi-label classification tasks, \my\ succeeds the \textit{similar classification accuracy} within a consistently \textit{shorter time} and with a \textit{relatively smaller compressed graph}. 
\begin{table*}

\centering{
\caption{ {Performance comparison of the multi-label classification tasks for the similarity threshold $\lambda=0.5$ and training ratio 5\% and 50\% for DBLP and BlogCatalog respectively}}
\label{table:compare}
\begin{tabular}{|c | c c c | c c c |}
 \hline 
  \multirow{2}{*}{$ $} & \multicolumn{3}{|c|}{\cellcolor{gray!60}\textbf{DBLP (5\%)}} & \multicolumn{3}{|c|}{\cellcolor{gray!60}\textbf{BlogCatalog (50\%)}} \\ \cline{2-7}
 
& \cellcolor{gray!25}\textbf{\my(DW)} & \cellcolor{gray!25}\textbf{DW} & \cellcolor{gray!25}\textbf{Gain} \% & \cellcolor{gray!25}\textbf{\my(DW)} & \cellcolor{gray!25}\textbf{DW} & \cellcolor{gray!25}\textbf{Gain} \% \\ \hline
\cellcolor{gray!25}\textbf{Macro} $\mathbf{F_1}$ &  0.625 & 0.622 & 0.51 & 0.243 & 0.250 & -2.92 \\
\cellcolor{gray!25}\textbf{Micro} $\mathbf{F_1}$ &  0.657 & 0.653 & 0.63 & 0.369 & 0.391 & -5.68 \\
\cellcolor{gray!25}\textbf{Time(s)} & 39.97 & 93.96 & \textbf{57.46} & 71.7 & 99.3 & \textbf{27.79}\\ \hline

& \cellcolor{gray!25}\textbf{\my(N2V)} & \cellcolor{gray!25}\textbf{N2V} & \cellcolor{gray!25}\textbf{Gain} \% & \cellcolor{gray!25}\textbf{\my(N2V)} & \cellcolor{gray!25}\textbf{N2V} & \cellcolor{gray!25}\textbf{Gain} \%\\ \hline
 \cellcolor{gray!25}\textbf{Macro} $\mathbf{F_1}$ & 0.626 & 0.625 & 0.13 & 0.251 & 0.262 & -4.11 \\
\cellcolor{gray!25}\textbf{Micro} $\mathbf{F_1}$ & 0.658 & 0.657 & 0.24 & 0.368 & 0.396 & -6.68 \\
\cellcolor{gray!25}\textbf{Time(s)} & 75.81 & 175.31 & \textbf{56.75} & 1247.14 & 1628.59 & \textbf{23.42}\\ \hline
 & \cellcolor{gray!25}\textbf{Compressed} & \cellcolor{gray!25}\textbf{Original}& \cellcolor{gray!25}\textbf{Gain} \%  & \cellcolor{gray!25}\textbf{Compressed} & \cellcolor{gray!25}\textbf{Original}& \cellcolor{gray!25}\textbf{Gain} \% \\ \hline
\cellcolor{gray!25}$\mathbf{|V|}$ &8824 & 32984 & \textbf{69.78}  & 8507 &	543872 & \textbf{17.50}\\
\cellcolor{gray!25}$\mathbf{|E|}$ & 32984	& 133664 & \textbf{75.32} &	10312 &	667966 & \textbf{18.58} \\ \hline
\end{tabular}}
\end{table*}

\vspace{3mm}
\begin{figure*}
    \centering
    \begin{subfigure}[h]{0.32\textwidth}
        \centering
        \includegraphics[width=1.12\textwidth]{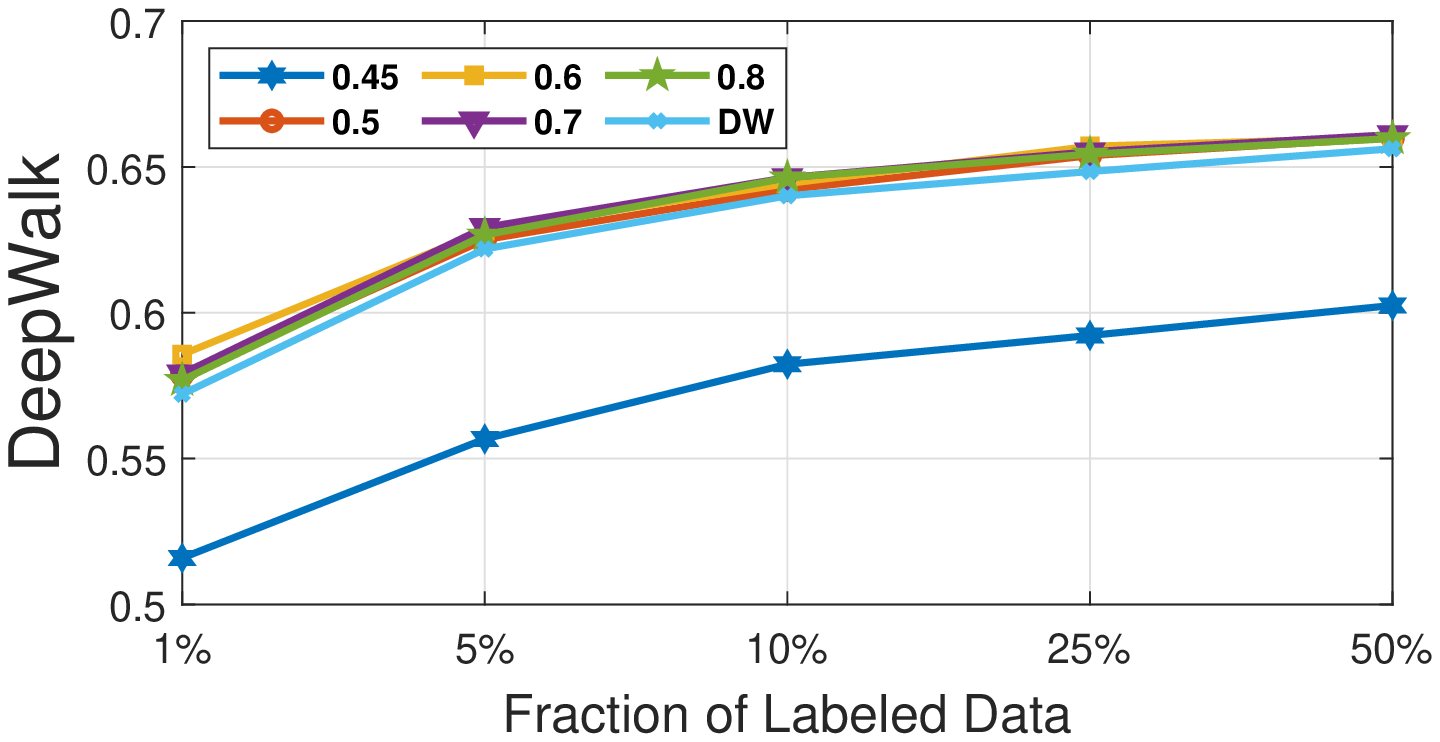}
    \end{subfigure}%
    ~ 
    \begin{subfigure}[h]{0.32\textwidth}
        \centering
        \includegraphics[width=1.12\textwidth]{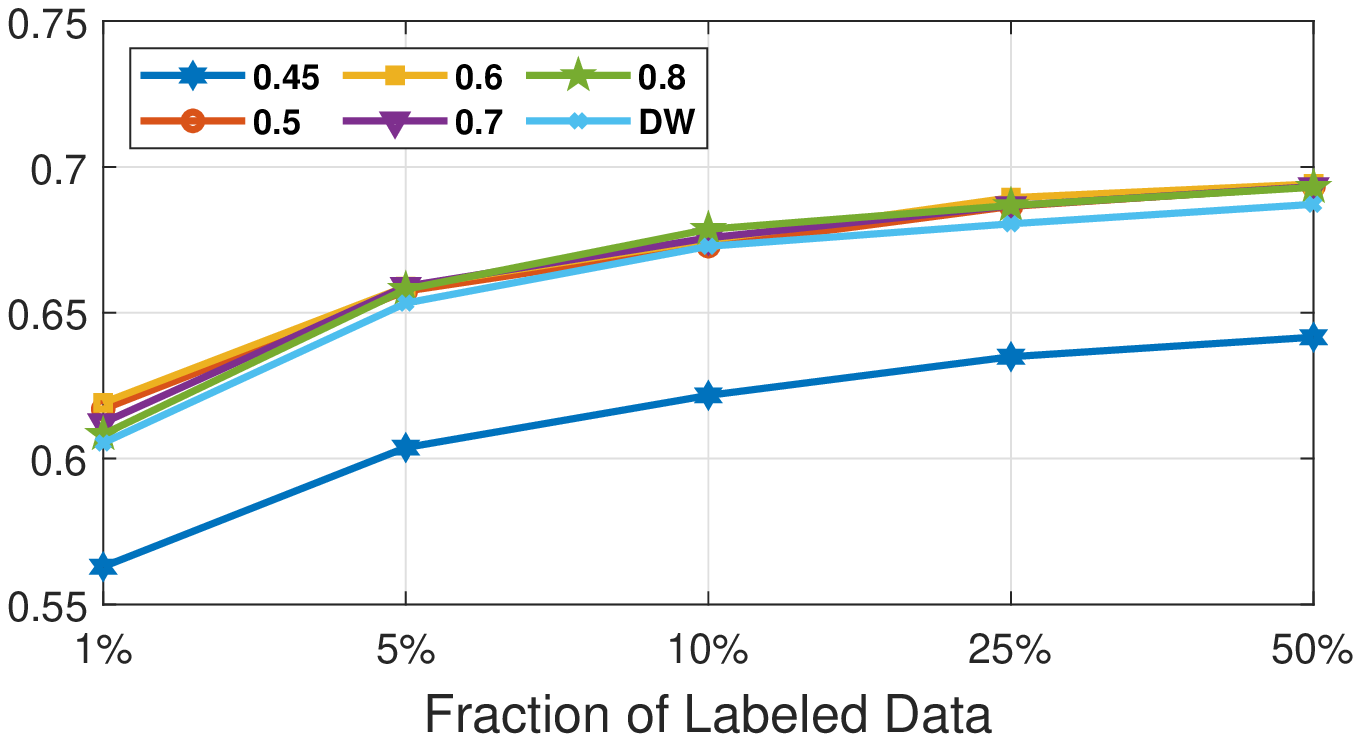}
    \end{subfigure}
        ~ 
    \begin{subfigure}[h]{0.33\textwidth}
        \centering
        \includegraphics[width=1.12\textwidth]{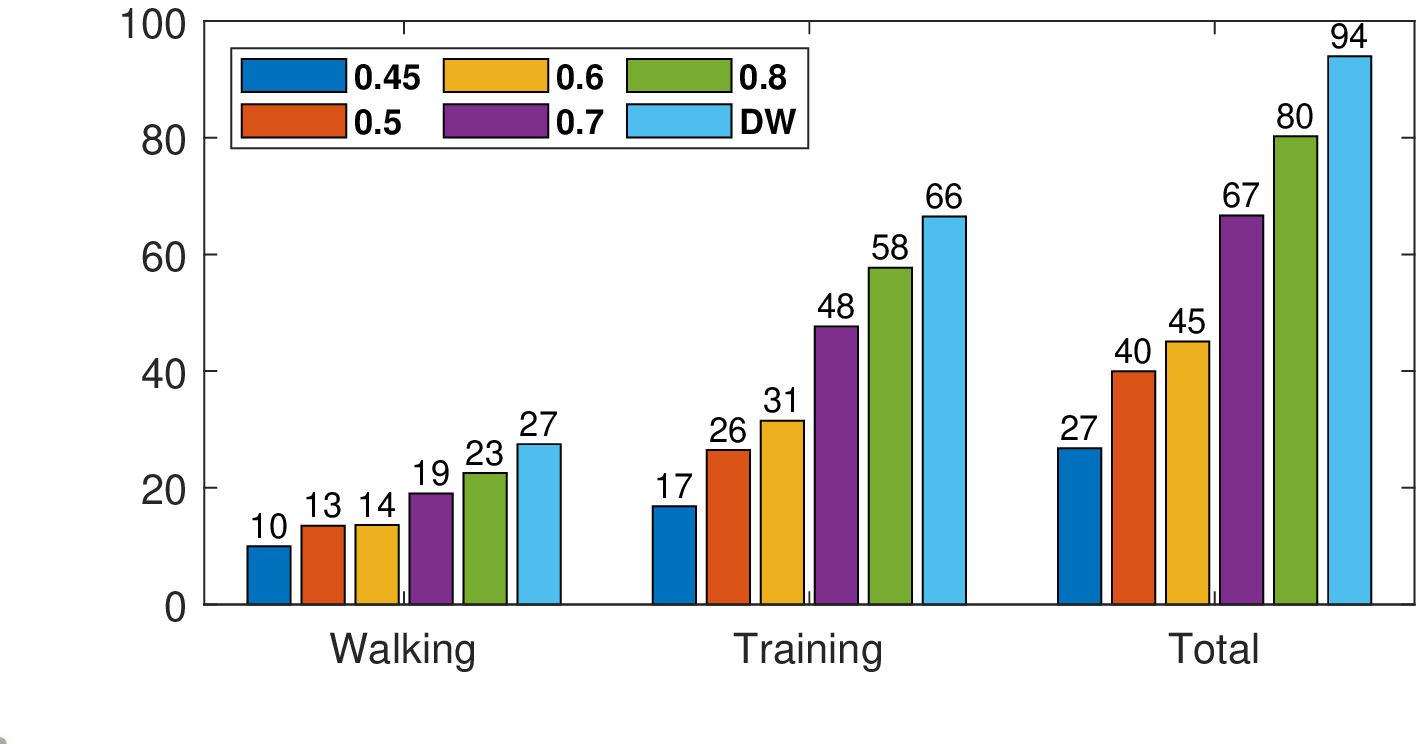}
    \end{subfigure}
~
    \begin{subfigure}[h]{0.32\textwidth}
        \centering
        \includegraphics[width=1.12\textwidth]{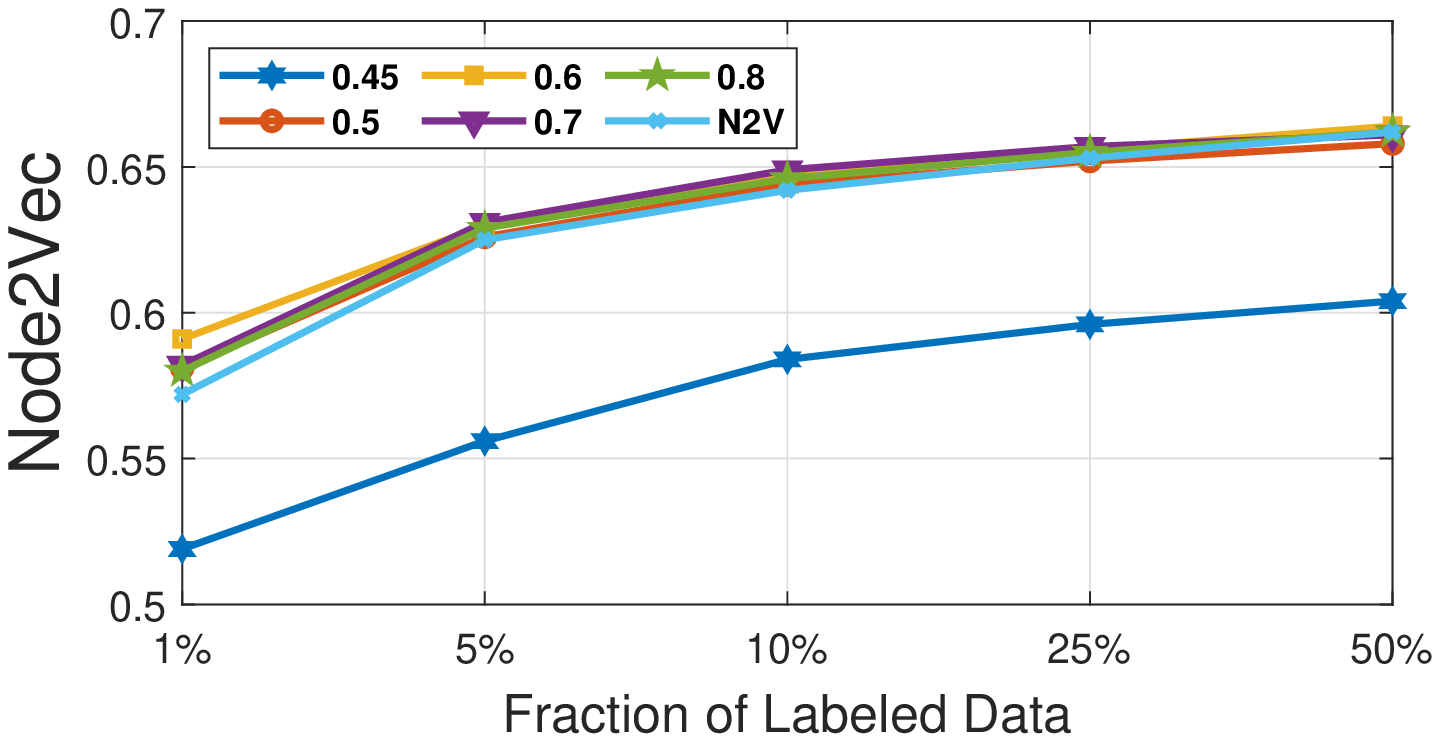}
        \caption{Macro $F_1$}
    \end{subfigure}%
    ~ 
    \begin{subfigure}[h]{0.32\textwidth}
        \centering
        \includegraphics[width=1.12\textwidth]{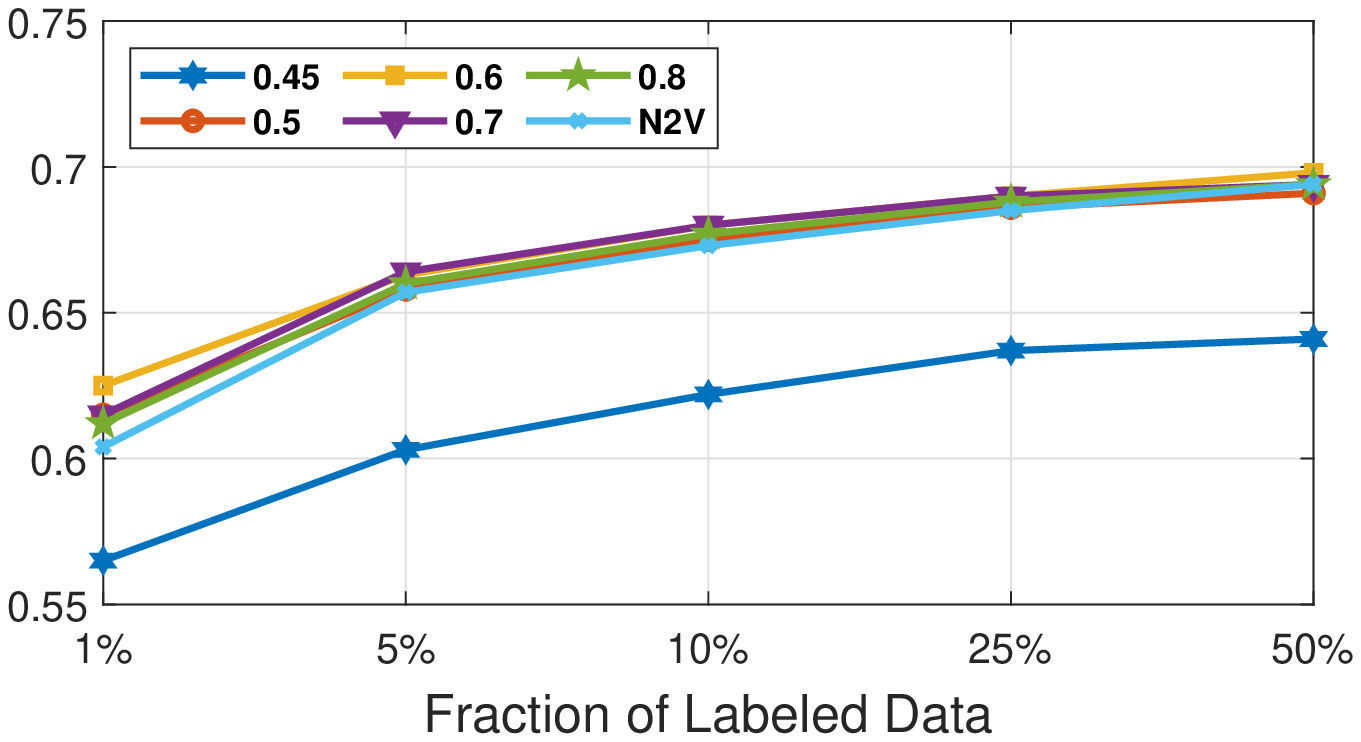}
        \caption{Micro $F_1$}
    \end{subfigure}
        ~ 
    \begin{subfigure}[h]{0.33\textwidth}
        \centering
        \includegraphics[width=1.12\textwidth]{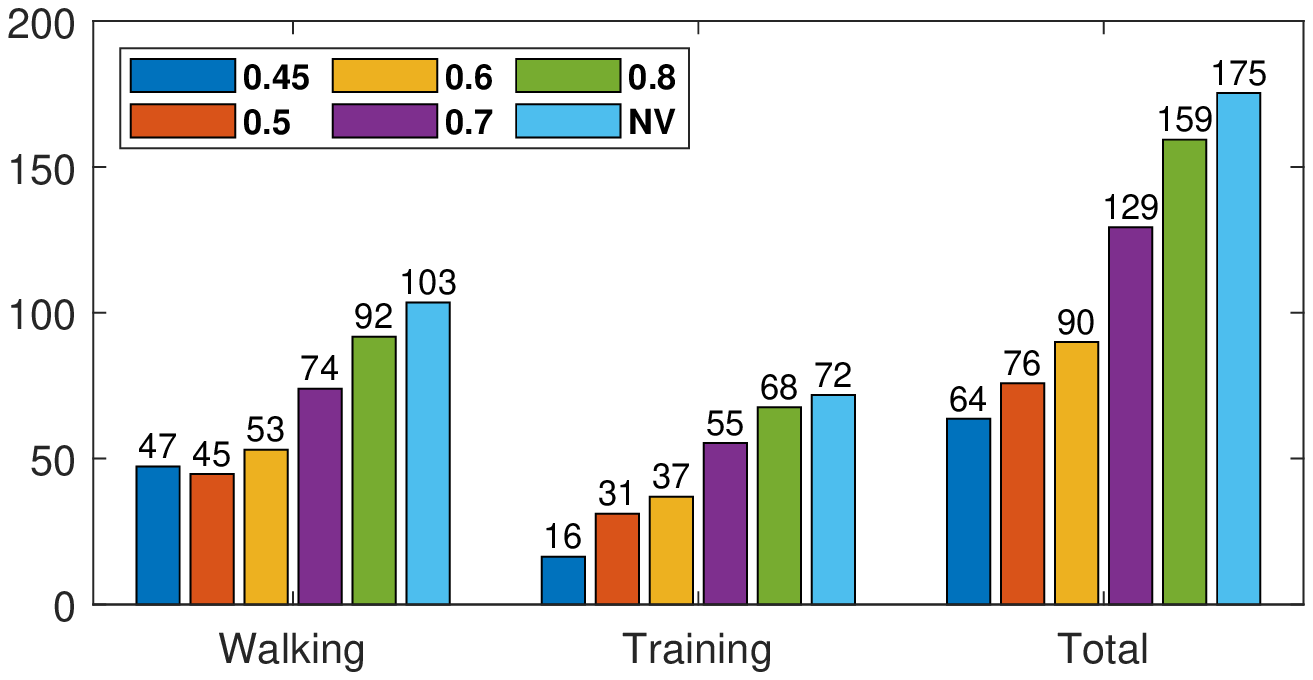}
        \caption{Running Time(s)}
    \end{subfigure}
    
    \caption{Detailed multi-label classification result on DBLP}
    \label{fig:db}
\end{figure*}

\begin{figure*}
    \centering
    \begin{subfigure}[h]{0.32\textwidth}
        \centering
        \includegraphics[width=1.12\textwidth]{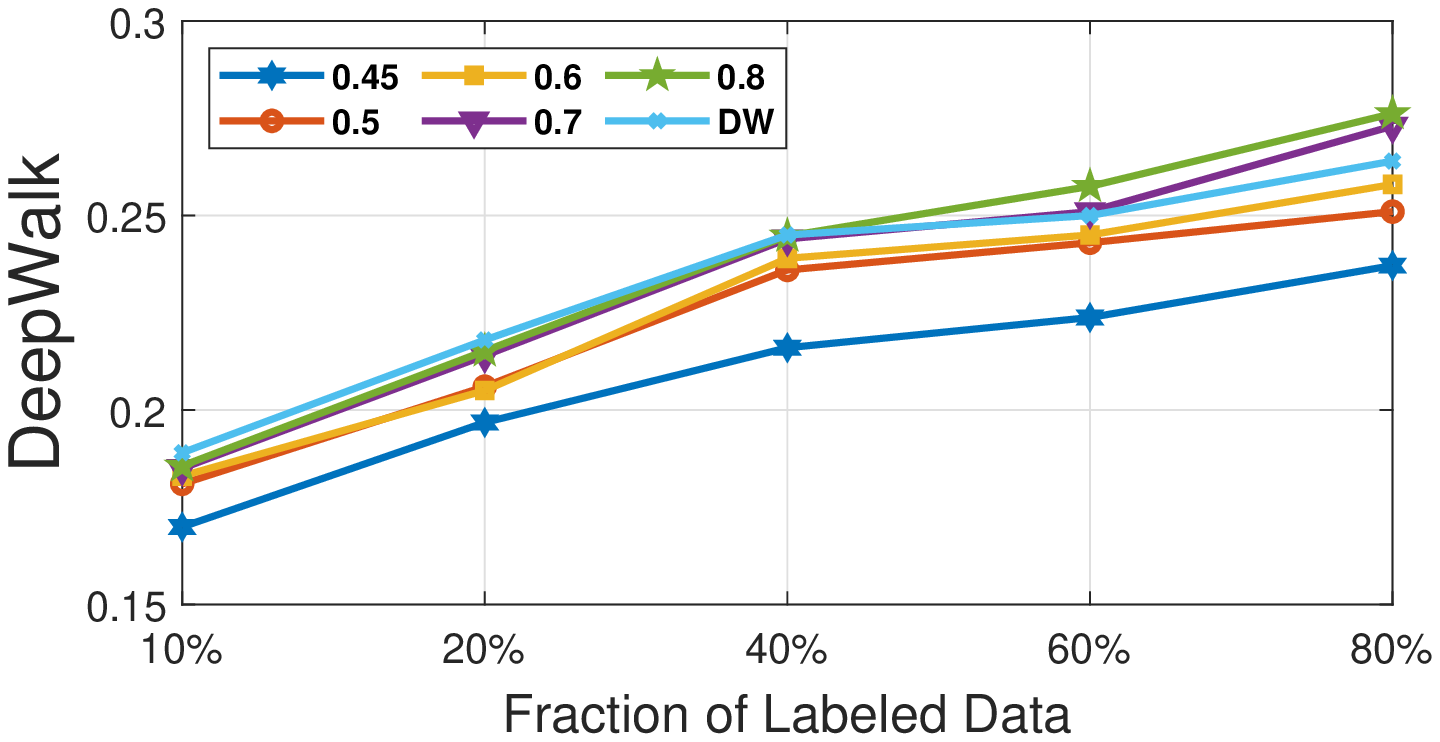}
    \end{subfigure}%
    ~ 
    \begin{subfigure}[h]{0.32\textwidth}
        \centering
        \includegraphics[width=1.12\textwidth]{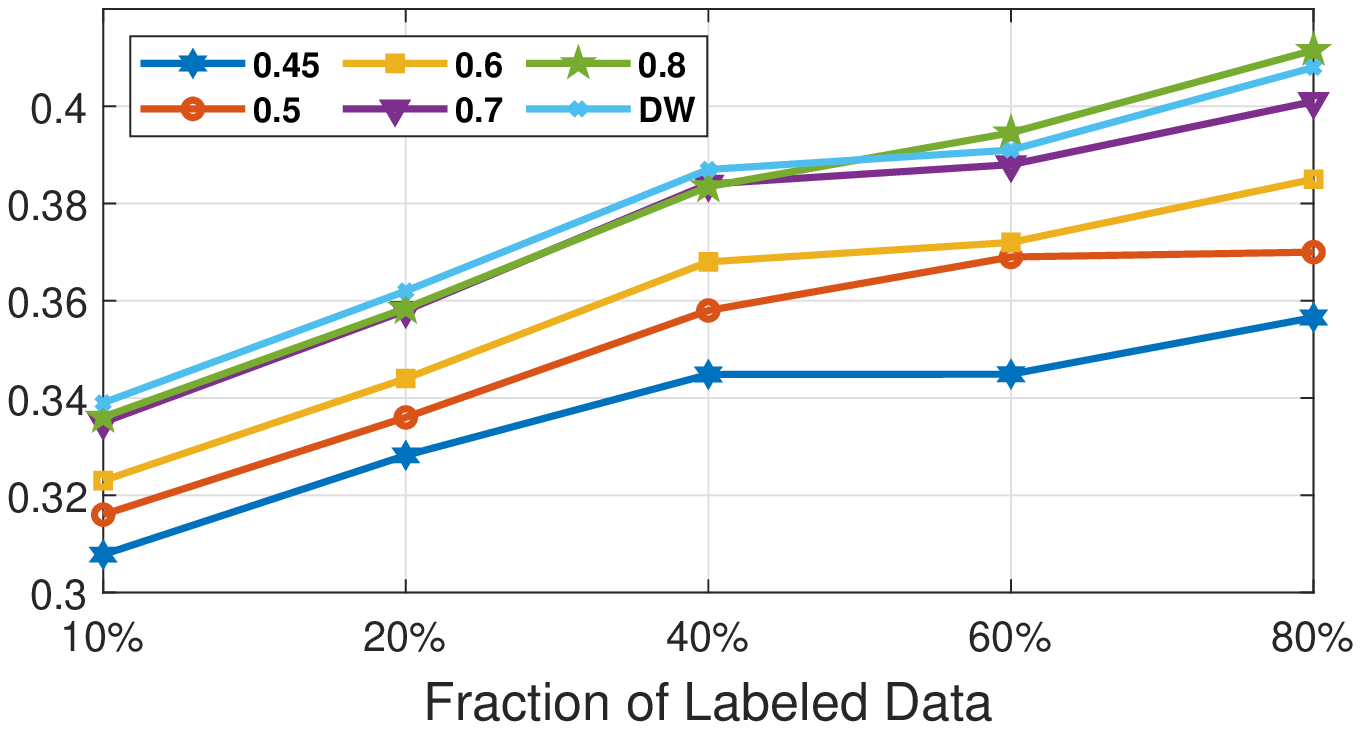}
    \end{subfigure}
        ~ 
    \begin{subfigure}[h]{0.33\textwidth}
        \centering
        \includegraphics[width=1.12\textwidth]{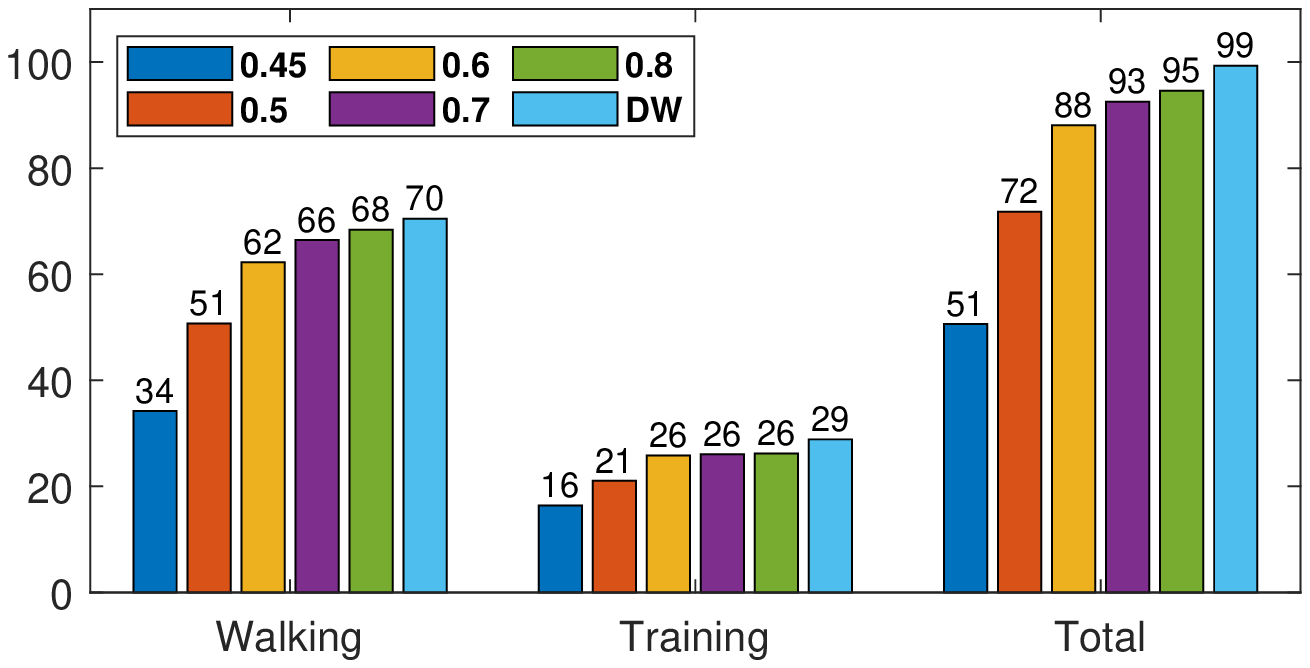}
    \end{subfigure}
~
    \begin{subfigure}[h]{0.32\textwidth}
        \centering
        \includegraphics[width=1.12\textwidth]{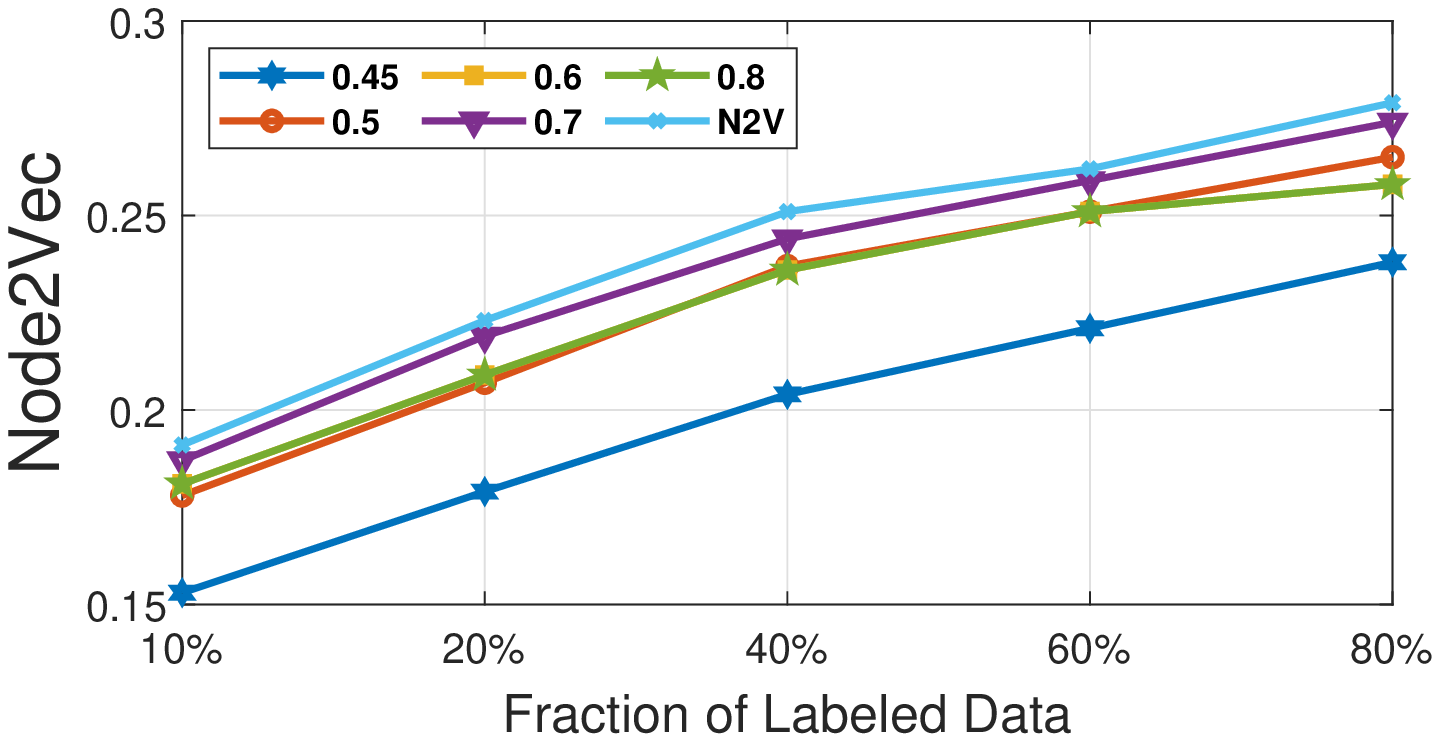}
        \caption{Macro $F_1$}
    \end{subfigure}%
    ~ 
    \begin{subfigure}[h]{0.32\textwidth}
        \centering
        \includegraphics[width=1.12\textwidth]{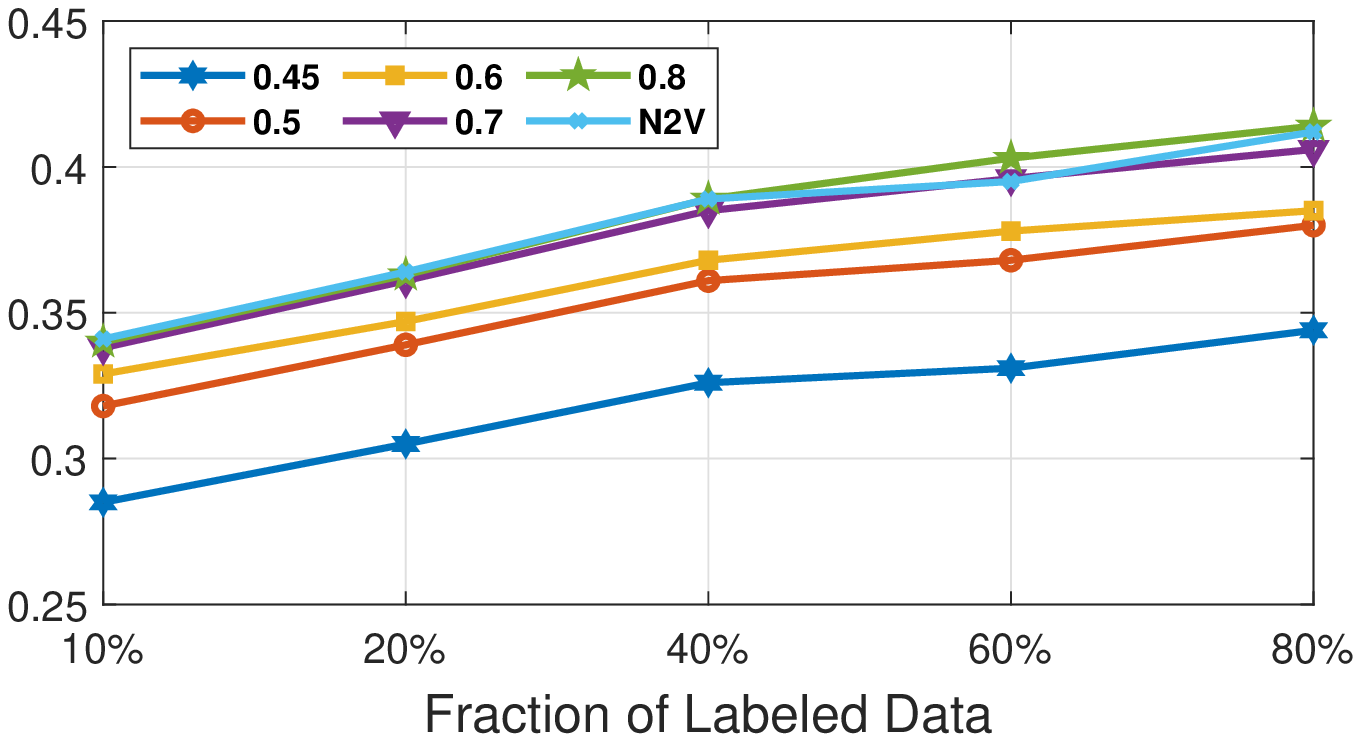}
        \caption{Micro $F_1$}
    \end{subfigure}
        ~ 
    \begin{subfigure}[h]{0.33\textwidth}
        \centering
        \includegraphics[width=1.12\textwidth]{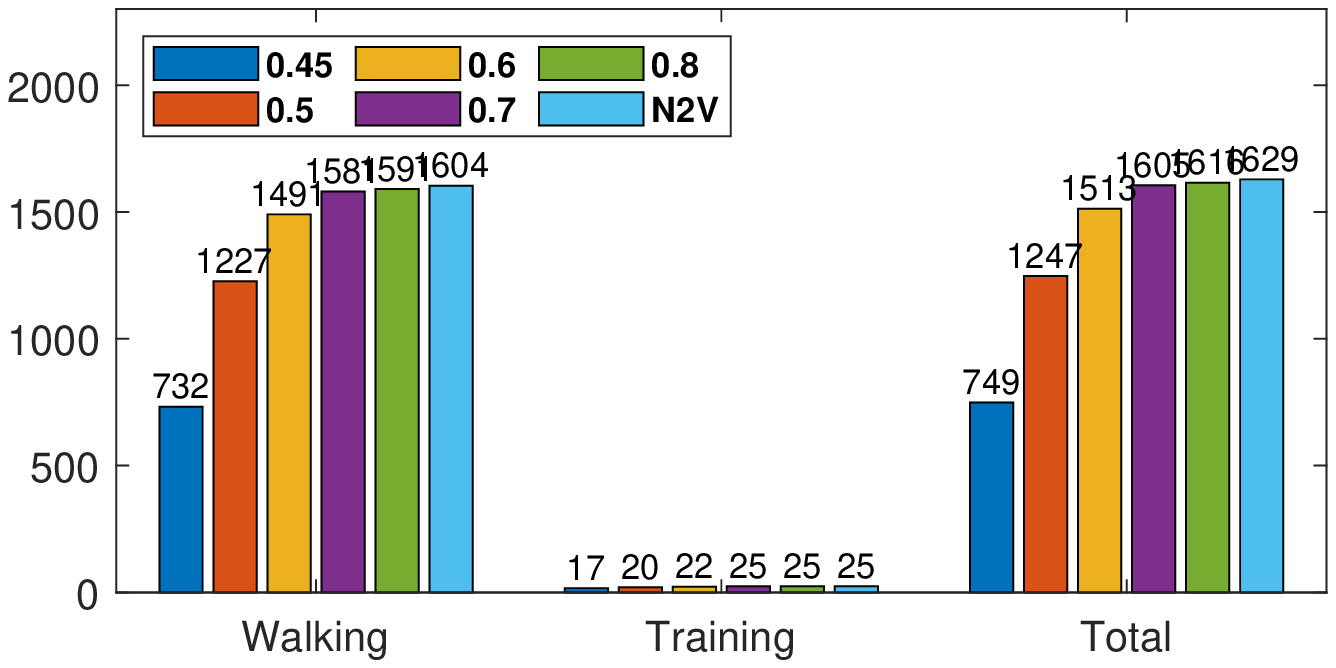}
        \caption{Running Time(s)}
    \end{subfigure}
    
    \caption{Detailed multi-label classification result on BlogCatolog}
    \label{fig:bc}
\end{figure*}

\subsection{Graph Compression}

\begin{figure*}[t!]
    \centering
    \begin{subfigure}[t]{0.23\textwidth}
        \centering
        \includegraphics[width=1.12\textwidth]{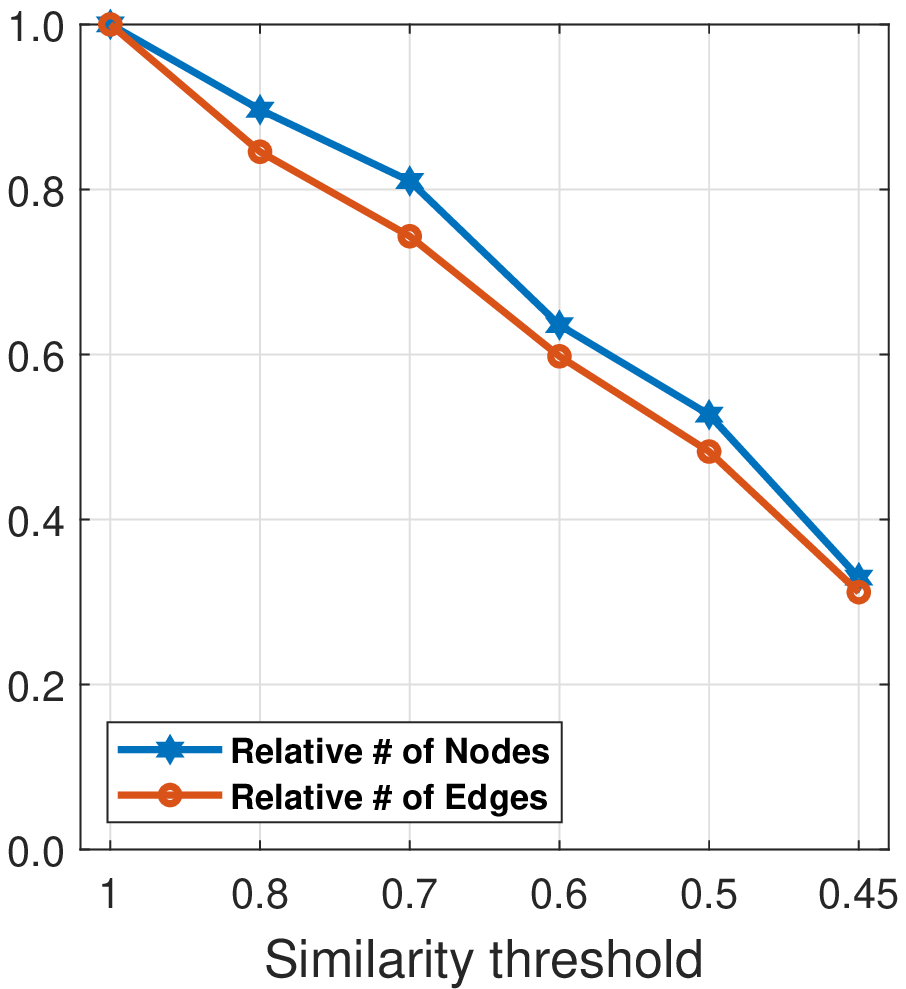}
        \caption{Cora}
    \end{subfigure}%
    ~ 
    \begin{subfigure}[t]{0.23\textwidth}
        \centering
        \includegraphics[width=1.12\textwidth]{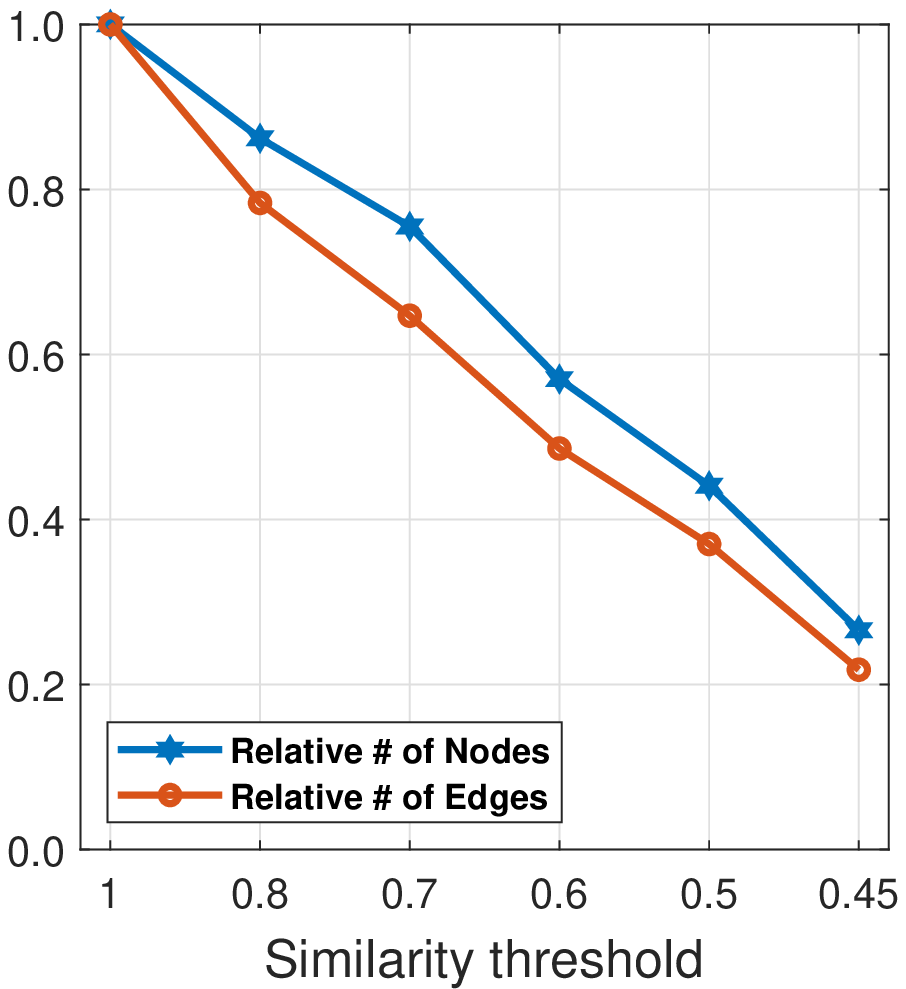}
        \caption{Wiki}
    \end{subfigure}
        ~ 
    \begin{subfigure}[t]{0.23\textwidth}
        \centering
        \includegraphics[width=1.12\textwidth]{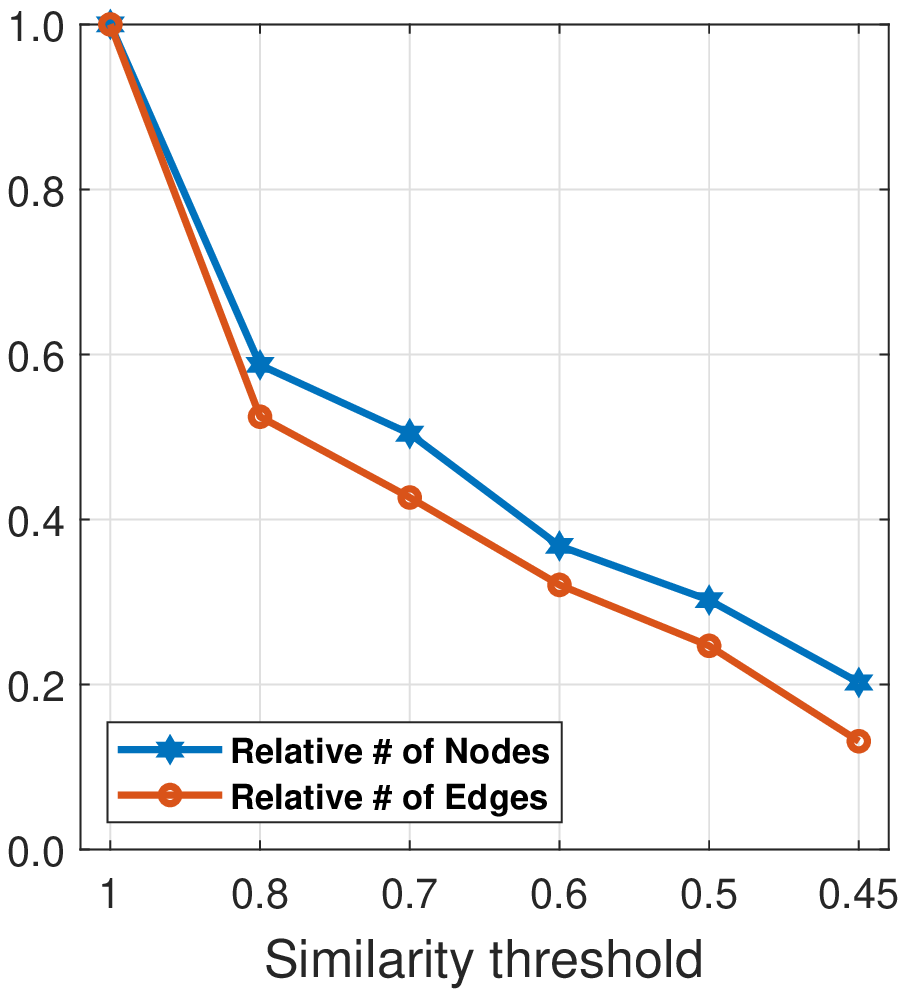}
        \caption{DBLP}
    \end{subfigure}
            ~ 
    \begin{subfigure}[t]{0.23\textwidth}
        \centering
        \includegraphics[width=1.12\textwidth]{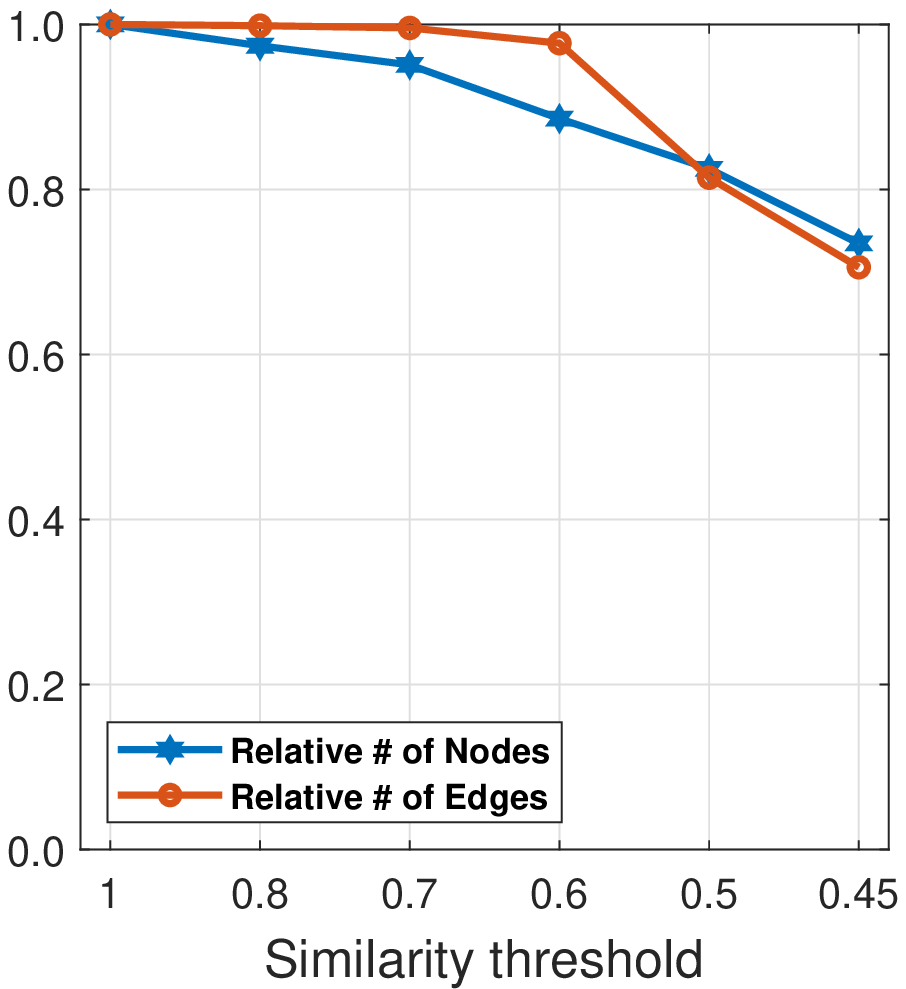}
        \caption{BlogCatalog}
    \end{subfigure}
    \caption{The ratio of vertices/edges of the compressed graphs to that of the original graphs.}
    \label{fig:cr}
\end{figure*}

In this section, we present how the graph size is decreased by compression with different similarity threshold values $\lambda$. As we see in Figure~\ref{fig:cr}, there is a linear relation between $\lambda$ and the number of vertices and edges till $\lambda=0.5$, and then graph sizes change dramatically for smaller $\lambda$ for Cora, Wiki and DBLP datasets, but the decrease is slow for BlogCatalog until $\lambda=0.7$. One of the possible reasons for BlogCatalog is the fact that the sizes of the neighbor sets for some vertices are very large, and it is not easy to get higher similarity for a larger set. For example for two vertices with 15 edges, 10 common neighbors can be considered to have higher similarity. On the other hand, two vertices with 150 edges, we should have 100 common neighbors to get the same similarity value which is not very common. 

\section{Related Work}
In this section, we briefly discuss the related work in the areas of networks embedding and graph compression.
\textbf{Network embedding.} Previous researchers consider the graph embedding as a dimensionality reduction~\cite{chen2018tutorial} such as PCA~\cite{wold1987principal} that captures linear structural information and LE (locally linear embeddings)~\cite{roweis2000nonlinear} that preserves the global structure of non-linear manifolds. While these methods are effective on small graphs, scalability is the major concern for them to be applied on large-scale networks with billions of vertices, since the time complexity of these methods is at least quadratic in the number of graph vertices~\cite{Zhang2017,wang2018graphgan}. On the other hand, recent approaches in graph representation learning focus on the scalable methods that use matrix factorization~\cite{qiu2018network} or neural networks~\cite{tang2015line,cao2016deep,Ying:2018:HGR}. Many of these aim to preserve the first and second order proximity as local neighborhood with path sampling using short random walks such as \dw\ and \ntv~\cite{hamilton2017representation,goyal2018graph,cai2018comprehensive,cui2018survey}. Some studies use network embedding on node and graph classification~\cite{deepwalk,chen2018harp,Niepert:2016:LCN}, some of them use it on graph clustering~\cite{Akbas2019,Akbas:2017:AGC,Cao:2015:grarep}.

\dw\ preserves the higher order proximity between vertices by generating random walks of fixed length from all the vertices of a graph. With considering the walks as sentences in a language model, they optimize the log-likelihood of random walks using the Skip-gram model~\cite{mikolov2013efficient}, which is for learning word embeddings. \dw\ uses hierarchical softmax for the efficiency of optimization. \ntv, which is from the many different extensions of \dw, makes an improvement to the random walk phase in \dw. They apply biased random walks using the return parameter $p$ and the in-out parameter $q$ to combine DFS-like and BFS-like neighborhood explorations. With this way, they preserve the network community and structural roles of vertices. Different than \dw, \ntv\ uses negative sampling for optimization.
 
Optimization in these methods could easily get stuck at a bad local minima as the result of poor initialization. Moreover, while preserving local proximities of vertices in a network, they may not preserve the global structure of the network. To address these issues, a multilevel graph representation learning paradigm, HARP, is proposed in \cite{chen2018harp} as a graph preprocessing step. In this approach, in a hierarchical manner at varying levels of coarseness, related vertices in the network are combined into super-nodes. After learning the embedding of the coarsened network with a state-of-the-art graph embedding method, the learned embedding is used as an initial value for the next level. In addition to capturing the global structure of the input graph by coalescing, by learning graph representation on these smaller graphs, a good initialization with the embedding of the coarsened network improves performance of  the state-of-the-art methods.

\my\ use the graph coarsening to capture the local structure of the network without hierarchical manner to improve the efficiency of the random walk based state-of-the-art methods.

\textbf{Graph compressing.} Although recent network embedding methods have a promising performance on the effectiveness of various applications, there are still some challenges since real-world graphs are massive in scale and this may obstruct the direct application of existing methods. On the other hand, when we consider a compressed or summary graph conserving the key structure and patterns of the original graph, many methods would be applicable to large graphs~\cite{Liu_2018}.

Graph compressing algorithms, which are popular
methods in the graph mining community, compress a graph into a smaller one with preserving certain properties of the original graph, such as connectivity ~\cite{zhou_comp}. Vertices with similar characteristics are grouped and represented by super-nodes. Approximations with compressing are used to solve the original problem more efficiently such as all-pairs shortest paths, search engine storage and retrieval~\cite{adler2001towards, suel_gc01}. Using an approximation of the original graph not only make a complex problem simpler but also make a good initialization to solve the problem. It has been proved successful in various graph theory problems~\cite{gilbert2004compressing}.

\my\ extends the idea of the graph compressing layout to network representation learning methods. We illustrate the utility of this paradigm by combining \my\ with two state-of-the-art representation learning methods, \dw\ and \ntv.

\section{CONCLUSIONS}
We propose a novel efficient network embedding method \my\ which preserves the local structural features of the vertices. To overcome the efficiency limitations of the state-of-the-art methods, we use the idea of the graph compressing layout to network representation learning methods. We combine related vertices of a network into super-nodes which preserve the neighborhood information of the vertices. Then, we use the compressed graph to learn the representation of the vertices in the original graph. We apply the utility of this paradigm by combining \my\ with two state-of-the-art representation learning methods, \dw\ and \ntv. Extensive experiments on a variety of different real-world graphs validate the efficiency of our approach on challenging multi-class and multi-label classification tasks without decreasing the effectiveness.

One of the future extensions of \my\ could be combining it with other kinds of graph representation learning methods which use matrix factorization and deep neural networks to see if it also works well with them. Another extension we are planing is using different similarity measures for compression to preserve different properties of the network. 

\bibliographystyle{abbrv}
\bibliography{vldb_sample}  

\end{document}